\documentclass[sigconf]{aamas} 

\usepackage{balance} 

\setcopyright{ifaamas}
\acmConference[AAMAS '23]{Proc.\@ of the 22nd International Conference
on Autonomous Agents and Multiagent Systems (AAMAS 2023)}{May 29 -- June 2, 2023}
{London, United Kingdom}{A.~Ricci, W.~Yeoh, N.~Agmon, B.~An (eds.)}
\copyrightyear{2023}
\acmYear{2023}
\acmDOI{}
\acmPrice{}
\acmISBN{}

\usepackage{algorithm}
\usepackage{algpseudocode}
\usepackage{tikz}
\usepackage{bm}
\usepackage{subcaption}
\usepackage{enumitem}

\newcommand{\todo}[1]{{\color{red} TODO: #1}}
\DeclareMathOperator*{\argmax}{arg\,max}

\DeclareMathOperator*{\NSW}{NSW}

\title[AAMAS-2023 Formatting Instructions]{Welfare and Fairness in Multi-objective Reinforcement Learning}

\author{Zimeng Fan$^*$}
\affiliation{
  \institution{Duke University}
  \city{Durham}
  \state{NC}
  \country{USA}}
\email{zf59@duke.edu}

\author{Nianli Peng$^*$}
\affiliation{
  \institution{Duke University}
  \city{Durham}
  \state{NC}
  \country{USA}}
\email{nianli.peng@duke.edu}

\author{Muhang Tian$^*$}
\affiliation{
  \institution{Duke University}
  \city{Durham}
  \state{NC}
  \country{USA}}
\email{muhang.tian@duke.edu}

\author{Brandon Fain}
\affiliation{
  \institution{Duke University}
  \city{Durham}
  \state{NC}
  \country{USA}}
\email{brandon.fain@duke.edu}

\begin{abstract}
We study fair multi-objective reinforcement learning in which an agent must learn a policy that simultaneously achieves high reward on multiple dimensions of a vector-valued reward. Motivated by the fair resource allocation literature, we model this as an expected welfare maximization problem, for some nonlinear fair welfare function of the vector of long-term cumulative rewards. One canonical example of such a function is the Nash Social Welfare, or geometric mean, the log transform of which is also known as the Proportional Fairness objective. We show that even approximately optimal optimization of the expected Nash Social Welfare is computationally intractable even in the tabular case. Nevertheless, we provide a novel adaptation of Q-learning that combines nonlinear scalarized learning updates and non-stationary action selection to learn effective policies for optimizing nonlinear welfare functions. We show that our algorithm is provably convergent, and we demonstrate experimentally that our approach outperforms techniques based on linear scalarization, mixtures of optimal linear scalarizations, or stationary action selection for the Nash Social Welfare Objective.
\end{abstract}



\newcommand{\BibTeX}{\rm B\kern-.05em{\sc i\kern-.025em b}\kern-.08em\TeX}

\begin{document}


\pagestyle{fancy}
\fancyhead{}


\maketitle 
\def\thefootnote{*}\footnotetext{These authors contributed equally to this work, ordered alphabetically by last name.}\def\thefootnote{\arabic{footnote}}

\section{Introduction}
\label{sec:intro}
Suppose a logistics company announces they will deploy a reinforcement learning agent that optimizes for completed deliveries. It works and the number of deliveries increase. Some days later, the company begins receiving complaints that delivery service to some locations has actually gotten much worse than before the AI deployment. The company assures its customers that the AI is learning and things will improve. But in another week, the situation is the same. Desperate, the company's engineers boost the reward weights associated with deliveries to those locations, only to find that now other locations are being neglected.

This is an example where data-driven algorithmic systems may be generally quite performant but nonetheless fail on structured subsets of input. This is well-known in reinforcement learning, where extensive ``reward shaping'' is sometimes necessary to achieve desired behavior. In the opening example, the desired behavior is a policy that achieves high delivery service rates at \textit{all} customer locations. However, standard reinforcement learning, in which the reward signal is a scalar value and the goal is to maximize total discounted reward, might naturally learn a policy that prioritizes ``easy'' to optimize regions (perhaps clusters of many tightly packed locations with many deliveries) at the expense of more difficult ways to achieve reward. Furthermore, because standard techniques rely on learning a \textit{stationary} policy, the policy continues to prioritize the same customers day after day. Addressing this problem may require problem-specific fine tuning of rewards, and even then ``fixing'' the original problem can introduce new ones.

In this paper, we take a different approach. We study nonlinear welfare optimization in the context of Multi-objective Reinforcement Learning (MORL). A Multi-Objective Markov Decision Process (MOMDP) is a Markov Decision Process where rewards are vectors instead of scalars. The components of this vector can be viewed either as different criteria like cost and time or, as as we interpret them, as individual utilities of \textit{``users''} to whom the learning agent should be \textit{fair}. In the opening example, the customers are the users and the vector reward tracks how well the learning agent is doing at optimizing deliveries for each user separately. 

The solution to a MOMDP is a policy that seeks to maximize some function of the cumulative reward vector. Due to the linearity of expectation, linear functions maximizing some weighted arithmetic mean of the cumulative reward vector are the simplest to use. However, for any particular selection of weights, the resulting policies may be undesirable from a fairness perspective as any linear function may ignore the utility of some users. For instance, for equal weights, a policy that gives user 1 a utility of 10 and user 2 a utility of 0 is preferred over another policy that yields a utility of 4 for both. We therefore study a more general class of \textit{welfare} functions with a particular emphasis on \textit{nonlinear} welfare functions that optimize for fairness and efficiency.

Optimizing a nonlinear welfare function in a MOMDP is a substantial algorithmic challenge. The Bellman optimality principles \cite{bellman1966dynamic}\cite{sutton1998introduction} no longer hold, and stationary policies, in which the action selection depends only on the current state and not history, are no longer necessarily optimal. This is quite intuitive for fairness. For example, if an AI personal assistant is tasked with grocery shopping for a household with competing preferences over desserts, it is relevant to the current decision whether one member of the household got their most or least favorite dessert in previous weeks, as the agent may wish to be fair to its users across time.

Though these examples are toys, there are many real-world decision problems in which a learning agent may need to simultaneously prioritize more than a single utility or goal in a balanced and fair way. In telecommunications and wireless networking, one may want to allocate bandwidth in a way that balances the quality of service in many different locations. In autonomous driving, one may want to balance vehicle speed and passenger comfort \cite{autonomous}.

\subsection{Contributions and Outline}
Sections~\ref{sec:related} and~\ref{sec:preliminaries} introduce related work and preliminaries for MOMDPs. Our results are as follows:
\begin{enumerate}[leftmargin=5mm, rightmargin=5mm, itemsep=1mm]
\item In Section~\ref{sec:desiderata} we introduce and characterize the problem of optimizing expected welfare for fairness in a MOMDP. We specifically focus on nonlinear welfare functions, with the Nash Social Welfare (NSW) as our canonical example of a fair welfare function.
\item In Section~\ref{sec:algorithm} we give a reduction in Theorem~\ref{theorem:hardness} to show that optimizing expected NSW is computationally intractable, even in the tabular setting. We further show that stationary policies cannot, in general, guarantee high approximation to optimality as the number of dimensions of reward grows.
\item On the positive side, also in Section~\ref{sec:algorithm} we define Algorithm~\ref{alg:cap} \textit{Welfare Q-Learning} that adapts model-free Q-learning in two important ways to optimize nonlinear welfare functions: (1) nonlinear learning updates, and (2) non-stationary action selection. We show in Theorem~\ref{theorem:convergence} that our algorithm is provably convergent.
\item In Section~\ref{sec:experiments} we deploy our algorithm in two simulated environments to optimize expected NSW. Our algorithm substantially outperforms the following baselines: (1) optimal [for NSW] linear scalarization, (2) optimal [for NSW] mixtures of optimal policies in each dimension, and (3) stationary action selection on our algorithm's learned Q-table.
\end{enumerate}



\section{Related Work}
\label{sec:related}
Multi-objective reinforcement learning (MORL) algorithms include \textit{single-policy} and \textit{multi-policy} methods \cite{liu2014multiobjective}. Single-policy methods use a scalarization function to reduce the problem to scalar optimization for a single policy. The simplest form is linear scalarization, applying a weighted sum on the Q vector \cite{moffaert2013hypervolume}. 

\textit{Multi-policy} methods search for a set of policies that approximate the Pareto frontier of the problem. For instance, the \textit{convex hull value-iteration algorithm} \cite{barrett2008learning} computes the deterministic stationary policies on the convex hull of the Pareto front. \textit{Pareto Q-learning} \cite{moffaert2013hypervolume} integrates temporal difference algorithms with Pareto dominance relations to learn a set of Pareto dominating policies. \textit{Stochastic mixture policy} \cite{vamplew2009constructing} combines multiple deterministic base policies with a convex combination, choosing a base policy with a given probability at the start of each episode. We focus on single-policy methods with nonlinear scalarization, as the size of the Pareto frontier may grow exponentially with the dimensionality of the problem, and because the Pareto frontier may not be well-approximated by its convex hull for nonlinear welfare functions.

Fairness in Reinforcement Learning has been recently considered, beginning with \cite{jabbari2017fairness} in a scalar setting. More directly related to our work, \cite{ICML20} investigated the (Deep) MORL problem of learning a fair policy to optimize the \textit{Generalized Gini Social Welfare function} using nonlinear scalarization. \cite{AAMAS22} studied a similar problem and considered maximizing concave welfare functions generally and Nash welfare specifically, showing an optimal approach to optimizing the welfare of expected rewards in the tabular setting and an extension to the function approximation setting. 

Our work differs from these in two major ways: (1) we seek to optimize the expected welfare, rather than welfare of expected rewards (see Section~\ref{sec:desiderata}) which is fundamentally more challenging computationally (see Section~\ref{sec:algorithm}), and (2) we learn a non-stationary policy, as stationary policies may be far from optimal for optimizing expected (nonlinear) welfare (see Section~\ref{sec:algorithm} and \cite{ogryczak2014fair}).

We formulate our consideration of welfare functions in Section~\ref{sec:desiderata} based on consideration from the resource allocation literature \cite{moulin2004fair}. Our canonical example of a fair welfare function, the Nash Social Welfare (NSW) derives from Nash's solution to the bargaining game \cite{nash1950bargaining} and its n-player extension \cite{luce1989games}. Its log transform is commonly known as the proportional fairness objective. More recent studies have shown NSW maximization provides outstanding fairness guarantees when allocating both divisible and indivisible goods \cite{caragiannis2019unreasonable}.




\section{Preliminaries}
\label{sec:preliminaries}
\noindent \textbf{Multi-objective Markov Decision Process.} A Multi-objective Markov Decision Process (MOMDP) consists of a finite set $\mathcal{S}$ of states, a starting state $s_1\in\mathcal{S}$,\footnote{In general we may have a distribution over starting states; we assume a single starting state for ease of exposition.} a finite set $\mathcal{A}$ of actions (we let $\mathcal{A}(s)$ denote the subset of actions available in state $s$), and probabilities $\mathcal{P}_{a, s, s'} \in [0, 1]$ that determine the probability of transitioning to state $s'$ from state $s$ after taking action $a$. Probabilities are normalized so that $\sum_{a \in \mathcal{A}(s), s'} \mathcal{P}_{a, s, s'} = 1$ for all $s$. We also have a reward function $\bm{R}(s,a): \mathcal{S} \times \mathcal{A} \rightarrow \mathbb{R}^n$  for taking action $a$ in state $s$.\footnote{For simplicity of exposition, we assume rewards are deterministic.} Some states may be \textit{terminal}, meaning they transition only to themselves and yield $\bm{0}$ reward. 

Each of the $n$ dimensions of the reward vector correspond to one of the multiple objectives that are to be maximized. At each time step $t$, the agent observes state $s_t \in \mathcal{S}$, takes action $a_t \in \mathcal{A}(s_t)$, and receives a reward vector $\bm{r}_{t} = \bm{R}(s_t, a_t) \in \mathbb{R}^n$. The environment, in turn, transitions into $s_{t+1}$ with probability $\mathcal{P}_{a_t,s_t,s_{t+1}}$. Where clear from context, we will often omit the subscript and simply write the immediate reward vector as $\bm{r}$.

A \textit{trajectory} is a sequence of state, action, reward tuples $\tau = (s_1, a_1, \bm{r}_1), (s_2, a_2, \bm{r}_2),...,(s_T, a_T, \bm{r}_{T})$. A trajectory that begins in the starting state  $s_1$ and ends in a terminal state defines an \textit{episode}. For a discounting factor $\gamma \in [0, 1)$, the \textit{discounted cumulative return} of a trajectory is the vector
$$\bm{G}(\tau) = \sum_{t=1}^{\infty} \gamma^{t-1} \bm{r}_t.$$

A \textit{stationary policy} is function $\pi(a \mid s): \mathcal{S} \times \mathcal{A} \rightarrow [0, 1]$ that forms a probability distribution such that $\sum_{a \in \mathcal{A}(s)} \pi(a \mid s) = 1$ for all $s$. Such a policy is \textit{stationary} since the probability with which an action is selected depends only on the current state. More generally, a policy (not necessarily stationary) is a function $\pi(a \mid \tau, s)$ that may additionally depend on a given trajectory (intuitively, the history prior to reaching state $s$). 

An \textit{action value function} is defined as the expected total reward starting from $s$, taking action $a$, and following policy $\pi$ thereafter:
$$\bm{q}_\pi(s,a):=\mathbb{E}_{\tau \sim \pi}\left[\,\sum_{k=0}^{\infty} \gamma^k \bm{r}_{t+k}\Bigm\vert s_t=s, a_t=a\right]\,.$$

The value function of a policy $\pi$ from state $s$ is defined by:
$$\bm{v}_\pi (s) := \mathbb{E}_{\tau \sim \pi}\left[\sum_{k=0}^{\infty} \gamma^{k} \bm{r}_{t+k} \Bigm\vert s_t=s\right].$$

Our algorithms will aim to solve the learning problem by finding an estimate of $\bm{q}_\pi(s,a)$ and $\bm{v}_\pi (s)$. We denote such estimates as $\bm{Q}_\pi(s,a)$ and $\bm{V}_\pi(s)$, respectively.



\section{Problem Formulation}
\label{sec:desiderata}
Our goal is to learn a policy $\pi$ that maximizes $\mathbb{E}_{\tau \sim \pi} \left[ \bm{G}(\tau) \right]$ in all dimensions. To make this optimization objective concrete, we must specify a \textit{scalarization} function $f: \mathbb{R}^n \rightarrow \mathbb{R}$. In fair reinforcement learning, we think of each of the $n$ dimensions of the reward vector $\bm{r}$ as corresponding to a distinct user to whom the learning agent wishes to be fair. The scalarization function can thus be thought of as a \textit{welfare} function $W$ over the users, and the learning agent is a welfare maximizer. For a given welfare function $W$, our goal is then to compute a policy $\pi^*$ that maximizes expected welfare:
\begin{equation} 
\pi^* = \argmax_{\pi} \mathbb{E}_{\tau \sim \pi} \bigl[W \bigl(\bm{G}(\tau) \bigr) \bigr].
\label{eq:expected_welfare}
\end{equation}

\subsection{Welfare Axioms} Here we describe desirable properties of welfare functions in terms of general outcomes as vectors $\bm{v}$. You can think of these vectors as possible discounted cumulative reward vectors corresponding to different policies. In the fair division literature \cite{moulin2004fair}, the most basic requirement of a welfare function is \textit{monotonicity}.

\begin{definition} $W$ satisfies \textbf{monotonicity} if and only if for all $\bm{v}$ and $\bm{v'}$ with $v_j \geq v'_j$ for all $j$ and $v_i > v_i'$ for some $i$, $W(\bm{v}) > W(\bm{v'})$.
\end{definition}

Intuitively, \textit{monotonicity} specifies that given all else equal, one prefers to increase a user's utility. 
Related is \textit{Pareto optimality}, that an outcome should be efficient in the sense of not being dominated by any other outcome. 
\begin{definition}
An outcome $\bm{v}$ satisfies \textbf{Pareto optimality} if there is no other outcome $\bm{v'}$ such that $v'_i \geq v_i$ for all users $i$, and at least one inequality is strict. $W$ satisfies Pareto optimality if any $\bm{v}$ (within some feasible space) that maximizes $W$ is Pareto optimal.
\end{definition}

Welfare functions should also satisfy \textit{symmetry}, or indifference towards permutations of the input \cite{moulin2004fair}. This is the most basic form of a fairness guarantee, that different users are treated similarly.

\begin{definition}
$W$ satisfies \textbf{Symmetry} if for every permutation $\sigma$ of its inputs $\bm{v}$, $W(\sigma(\bm{v}))=W(\bm{v})$.
\end{definition}

A family of welfare functions satisfying the above properties are \textit{generalized mean $p$-welfare functions} \cite{barman22, moulin2004fair}, where $W_p(\bm{v}) = 1/n \left(\sum_{i=1}^{n} \left( v_i \right)^{p} \right)^{1/p}$. For instance, when $p = 1$, we have the utilitarian welfare function \cite{binmore1998egalitarianism}, the arithmetic mean of utilities.


Note that the utilitarian social welfare that may not be suitable for ensuring fairness on outcomes. Monotonicity, Pareto optimality, and Symmetry are merely minimal requirements for a welfare function. One may wish to introduce a stronger axiom such as the \textit{Pigou-Dalton Principle} \cite{castagnoli1990note}. This principle states that one-to-one transfer of utility (or rewards in the MOMDP) from a better-off user to a worse-off user should increase the overall welfare. 

\begin{definition}
$W$ satisfies the \textbf{Pigou-Dalton Principle} if for all $\bm{v}, \bm{v'}$ equal except for $v_i = v_i' + \delta$ and $v_j = v_j' - \delta$ where $v_j' - v_i' > \delta > 0 $, $W(\bm{v}) > W(\bm{v'})$.
\end{definition}

In other words, more equal distributions of utility are preferred. Functions that satisfy this formulation of fairness are often concave, capturing the diminishing marginal returns of increasing the utility of a user who already enjoys high utility relative to other users. Among the generalized mean $p$-welfare functions, the Pigou-Dalton Principle is satisfied for all $p < 1$. Our algorithm is designed to maximize welfare functions in this class. 


\subsection{Nash Social Welfare Function} The extreme case of a fair welfare function is the generalized mean $p$-welfare function where $p \rightarrow -\infty$, which corresponds to the \textit{egalitarian} welfare function \cite{sen2018collective} that maximizes the minimum utility (and subject to that, optimizes the next smallest, and so forth). In-between the extremes of the utilitarian and egalitarian social welfare functions, we specifically focus on the \textbf{Nash Social Welfare} (NSW) function as our canonical example of a fair welfare function that also balances efficiency with fairness \cite{nash1950bargaining, kaneko1979nash, fain2018fair, caragiannis2019unreasonable}.
\begin{equation}
    \NSW(\bm{v}) = \left(\prod_{i=1}^{n} v_i \right)^\frac{1}{n}
\end{equation}

NSW is also simply the geometric mean of utilities, and in its log transform is also known as the proportional fairness objective. Note that NSW is a generalized mean $p$-welfare function where $p \rightarrow 0$. In addition to the previous desirable properties, NSW also enjoys the property of being \textit{scale invariant}, meaning that the $\argmax$ of NSW is invariant under scaling of a given dimension of reward. From a practical perspective, this means that the relative scales of utility or reward for each dimension are not significant and do not need to be tuned during reward shaping.

Though we focus on NSW as our canonical example, we note that other reasonable welfare functions exist. For example, in Multi-Objective optimization, several works have studied Ordered Weighted Average (OWA) operators, a family of operators that contains many types of means \cite{YAGER1993125}.




\subsection{Expected Welfare} In contrast to some prior work \cite{ICML20, AAMAS22} we focus on optimizing the expected welfare $\mathbb{E}_{\tau \sim \pi} \bigl[W \bigl(\bm{G}(\tau) \bigr) \bigr]$ rather than the welfare of the expectation $W\bigl(\mathbb{E}_{\tau \sim \pi} [ \bm{G}(\tau) ]\bigr)$. Note that for any concave welfare function, including $\NSW$, \textit{Jensen's inequality} \cite{durrett2019probability} implies 

\begin{equation}
\mathbb{E}_{\tau \sim \pi} \bigl[W \bigl(\bm{G}(\tau) \bigr) \bigr] \leq W\bigl(\mathbb{E}_{\tau \sim \pi} [ \bm{G}(\tau) ]\bigr).
\end{equation}



We optimize for the lower bound (which turns out to be a more computationally challenging objective, see Section~\ref{sec:algorithm}) in order to avoid treating policies as ``fair'' that are unfair in every particular episode and satisfy fairness only across several episodes on average.


\paragraph{Example (Expected Welfare).} Consider the example diagrammed in Figure~\ref{fig:expectation} with $n=2$ users. Suppose we want to learn a policy that maximizes NSW. There is a stochastic policy $\pi_1$ that yields discounted cumulative reward of $(1,0)$ with probability $0.5$ and $(0,1)$ with probability $0.5$. There is also a deterministic policy $\pi_2$ that yields $(0.5-\epsilon, 0.5-\epsilon)$ (where $\epsilon>0$ is small). 
The NSW expected reward under $\pi_1$ is 0.5, even though with probability 1, the NSW of every trajectory generated by $\pi_1$ is 0. By contrast, the NSW of $\pi_2$ is always $0.5-\epsilon$.  Maximizing the expected welfare, our optimization problem would prefer $\pi_2$. 

\begin{figure}[h]
    \centering
    \begin{tikzpicture}[main/.style = {draw, circle}, node distance={60mm}] 
    \node[main] (1) at (0,0) {$s_1$}; 
    \node[main] (2) at (5,0) {$s_2$};
    \draw[->, thick] (1) to node[midway, above]{$(0.5-\epsilon, 0.5-\epsilon)$}(2);
    \draw[->, thick, dotted] (1) to [out=-30, in=-150, looseness=1]node[midway, above]{$(0,1)$}(2);
    \draw[->, thick, dotted] (1) to [out=30,in=150,looseness=1]node[midway, above] {$(1,0)$}(2);
    \end{tikzpicture} 
    \caption{Example MOMDP. Dotted lines represent trajectories generated by $\pi_1$, solid line for $\pi_2$}
    \label{fig:expectation}
\end{figure}
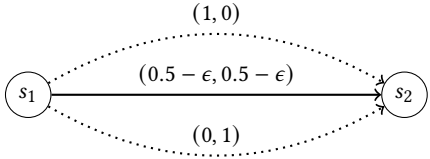

This example shows the intuition for why we choose to maximize $\mathbb{E}_{\tau \sim \pi} \bigl[W \bigl(\bm{G}(\tau) \bigr) \bigr]$. We seek to find a policy that generates trajectories with high expected welfare, a stronger property than generating high welfare of expected rewards. As we see in the next section, the problem is also computationally more challenging. 


\section{Optimizing Welfare}
\label{sec:algorithm}
In general, one cannot provably and efficiently optimize all fair welfare functions. We first demonstrate that finding a policy that maximizes the $\NSW$ is APX-hard, implying one cannot get an arbitrarily close approximation efficiently, even in the tabular setting. We note that the same is not true for optimizing the NSW of expected rewards, for which the optimal stochastic policy can be computed efficiently \cite{AAMAS22}. 

Our argument follows via a reduction from the problem of allocating indivisible goods, in which $m$ items must be partitioned among $n$ users, where user $a_i$ has utility $u_{i, j} \geq 0$ for good $j$ and their utility for multiple goods is the sum of their utilities for the individual goods.

\begin{lemma} 
\cite{Hardness17} It is APX-hard to compute an indivisible allocation of goods optimizing the NSW. 
 \end{lemma}

From this we can show the following impossibility.
\begin{theorem}
\label{theorem:hardness}
Computing the policy that maximizes $\NSW \bigl(\bm{G}(\tau) \bigr)$ is APX-hard, even in a deterministic environment.
\end{theorem}
\begin{proof}
We reduce the problem of allocating a set of indivisible items to users with additive utilities. Given such an allocation problem, consider the following MOMDP. Find an arbitrary enumeration $\{s_1, s_2, \dots, s_m\}$ of the $m$ items. These are the states of the MOMDP. At each time step, the environment transitions from $s_j$ to $s_{j+1}$, beginning at $s_1$ and with $s_m$ as a terminal state. In each state, $s_j$, the agent has $n$ available actions $\{a_1,\dots,a_n\}$ where taking action $a_i$ in state $s_j$ corresponds to allocating item $s_j$ to user $a_i$ and receives a reward vector with $u_{i, j}$ in dimension $i$ and $0$ elsewhere. 


Consider a policy $\pi$ that is $\mu$-approximately optimal on the $NSW$ objective. The indivisible allocation corresponding to each user receiving the set of goods for which $\pi$ chooses $a_i$ is also $\mu$-approximately optimal on the indivisible allocation problem. By Lemma 1, there exists a constant $\mu > 1$ such that it is NP-hard to approximate the NSW optimal allocation of indivisible items to agents with additive utilities within a factor $\mu$. It must also be NP-hard to approximate the NSW optimal policy in a MOMDP within a factor $\mu$.

\end{proof}

\subsection{Non-stationary Welfare Q-Learning} 
We now present our algorithm, \textit{Welfare Q-Learning}, which implements a variant on \textit{Q-learning} \cite{watkins1989learning}, a model-free temporal-difference learning algorithm \cite{sutton1988learning}. Our algorithm differs in two major ways from standard Q-learning.

\begin{enumerate}
    \item Q table updates are chosen to maximize the (potentially) nonlinear $W$, and each value $\bm{Q}(s, a)$ is a vector in $\mathbb{R}^n$ corresponding to an estimate of the future reward vector possible that is welfare maximal.
    \item Action selection is non-stationary. We keep track of the discounted cumulative reward vector within a trajectory so far and select the action that maximizes total estimated welfare including that already accumulated and future estimates.
\end{enumerate}

\begin{algorithm}
\caption{\textit{Welfare Q-Learning}}\label{alg:cap}
\begin{algorithmic}[1]
\State \textbf{Parameters:} Learning rate $\alpha \in (0, 1]$, Discount factor $\gamma \in [0, 1)$, exploration rate $\epsilon > 0$, welfare function $W$
\State \textbf{Require:} Initialize $\bm{\bm{Q}}(s, a)$ for all $s \in \mathcal{S}, a \in \mathcal{A}(s)$ arbitrarily except $\bm{\bm{Q}}(\bar{s}, \cdot) \gets 0$ for terminal states $\bar{s}$
\For{each episode}
\State Initialize $s \gets s_1$, $\bm{r_{acc}} \gets \bm{0}$, $c\gets0$

\Repeat \Comment{each step in an episode}
\State$$a \leftarrow
\begin{cases}
\text{a uniform random action} & \text{with} \Pr(\epsilon) \\
\argmax_{a'} W \left( \bm{r_{acc}} + \gamma^c\bm{\bm{Q}}(s,a') \right) & \text{otherwise}
\end{cases}$$
\State Take action $a$, observe $\bm{r}, s'$
\State $a^* \gets \argmax_a W[\gamma \bm{\bm{Q}}(s', a)]$
\State $\bm{\bm{Q}}(s, a) \leftarrow \bm{\bm{Q}}(s, a) + \alpha [\bm{r} + \gamma \bm{\bm{Q}}(s', a^*) - \bm{\bm{Q}}(s, a)]$
\State $s \leftarrow s'$
\State $\bm{r_{acc}} \gets \bm{r_{acc}} + \gamma^c \bm{r}$
\State $c \gets c+1$
\Until{$s$ is terminal}
\EndFor
\end{algorithmic}
\end{algorithm}


We show experimentally in Section~\ref{sec:experiments} that both of these changes are crucial to achieving high expected welfare for fair welfare functions such as NSW. We show in Theorem~\ref{theorem:convergence} that the algorithm still provably converges even with the nonlinear learning updates on a vector-valued Q table. To see the intuition for the significance of non-stationary action selection for optimizing the expected welfare of a nonlinear welfare function, we present the following example.

\subsubsection{Non-stationarity}
Consider the MOMDP diagrammed in Figure~\ref{fig:nonstationary}. There are $n$ users. $|\mathcal{A}(s_1)| = 1$, and at $t = 1$, the environment transitions from $s_1$ to $s_{2,i}$ with probability $1/n$ and reward $\bm{0} \in \mathbb{R}^n$. At $t = 2$, the environment transitions from $s_{2,i}$ to $s_3$ with probability $1$ and a reward vector that has $1$ at the $i^{th}$ component and $0$ elsewhere. At $t = 3$, the agent gets to choose from $n$ actions, $a_1, a_2, \dots, a_n$ that yield associated with rewards $(1,\dots,1,0), (1,\dots,0,1), \dots, (0, \dots, 1,1)$, respectively. 

A stationary policy cannot achieve expected NSW greater than $1/n$ on this MOMDP, even though a non-stationary policy can achieve expected NSW of 1. Without loss of generality, assume a stationary stochastic policy $\pi$ chooses the action $a_i$ with probability $p_i$ for $i = 1, \dots, n$ such that $\sum_i p_i = 1$. Then the expected Nash social welfare resulting from such a policy is
$$
\mathbb{E}_{\tau \sim \pi}\left[\NSW\left[\sum_t \bm{r_t}\right]\right] = \sum_{i=1}^n \frac1n p_i = \frac1n
$$

On the other hand, an optimal non-stationary policy $\pi^*$ that keeps track of accumulated rewards is able to choose the correct complementary action at $s_3$ depending on the random transition at $t = 1$, so that $\mathbb{E}_{\tau \sim \pi^*}\left[\NSW\left[\sum_t \bm{r_t}\right]\right] = 1$.

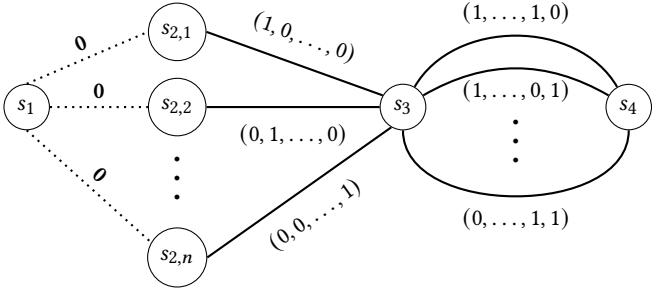
\begin{figure}[h]
    \centering
    \begin{tikzpicture}[main/.style = {draw, circle}, node distance={40mm}] 
    \node[main] (1) at (-1,0) {$s_1$}; 
    \node[main] (2) at (1,1) {$s_{2,1}$};
    \node[main] (4) at (1,0) {$s_{2,2}$};
    \node[main] (5) at (1,-2) {$s_{2,n}$};
    \path (4) -- (5) node [black, font=\Huge, midway, sloped] {$\dots$};
    \draw[-, thick, dotted] (1) to [out=0, in=180, looseness=0]node[midway, above, sloped]{$\bm{0}$}(4);
    \draw[-, thick, dotted] (1) to [out=90,in=180,looseness=0]node[midway, above, sloped] {$\bm{0}$}(2);
    \draw[-, thick, dotted] (1) to [out=270,in=150,looseness=0]node[midway, above, sloped] {$\bm{0}$}(5);
    \node[main] (6) at (4,0) {$s_3$}; 
    \draw[-, thick] (2) to [out=0, in=150, looseness=0]node[midway, above=3, sloped]{$(1,0, \dots, 0)$}(6);
    \draw[-, thick] (4) to [out=0,in=180,looseness=0]node[midway, below=3]{$(0, 1, \dots, 0)$}(6);
    \draw[-, thick] (5) to [out=0,in=-120,looseness=0]node[midway, below=3, sloped]{$(0, 0, \dots, 1)$}(6);
    \node[main] (3) at (7,0) {$s_4$};
    \draw[-, thick] (6) to [out=60, in=120, looseness=1]node[midway, above=1, sloped]{$(1, \dots, 1, 0)$}(3);
    \draw[-, thick] (6) to [out=30, in=150, looseness=1]node[midway, below=1, sloped]{$(1, \dots, 0, 1)$}(3);
    \path (5.5,0) -- (5.5,-1) node [black, font=\Huge, midway, sloped] {$\dots$};
    \draw[-, thick] (6) to [out=-90, in=-90, looseness=1]node[midway, below=1, sloped]{$(0, \dots, 1, 1)$}(3);
    \end{tikzpicture} 
    \caption{Example of a stochastic MOMDP in which the expected Nash social welfare of stationary policies shrinks to $0$ as the dimension of rewards $n$ increases.}
    \label{fig:nonstationary}
\end{figure}

This example shows why, when optimizing for fairness, one should keep track of the discounted cumulative reward $\bm{r_{acc}}$ allocated to thus far. A stationary policy does not distinguish which users have received higher or lower rewards on a given trajectory thus far. Our approach seeks to be greedy on the sum of this discounted cumulative reward and future estimates of reward stored in the Q table. In this manner, users that have not received much reward so far within an episode are prioritized in action selection.


It is worth noting that in order to incorporate both information from the past and future, we need a consistent accounting of discounting for both terms to ensure our agent has correct information from both past an future when deciding a fair policy:
\begin{align*}
    \bm{r_{acc_t}} &= \bm{r}_1 + \gamma \bm{r}_2 + \gamma^2 \bm{r}_3 + ... + \gamma^{t-1} \bm{r}_t\\
    \gamma^t \bm{\bm{Q}}_\pi(s,a) &= \mathbb{E}_\pi \left[\gamma^t\bm{r}_{t+1}+\gamma^{t+1}\bm{r}_{t+2}+... \mid S_t=s, A_t=a \right].
\end{align*}

It is also important that each value in our learned Q table stores the vector of estimated future rewards rather than simply a scalar estimate of the welfare achievable, as the true greedy objective is the welfare of the sum of these vectors, not sum of the welfare of the two. Because of this, our Q table has size $|\mathcal{S}| \times |\mathcal{A}| \times n$ where $n$ is the dimension of the reward. Indeed, as we show in Section~\ref{sec:experiments}, this does result in decreased convergence rates of the algorithm for larger $n$.


\subsubsection{Convergence}
We now argue that the Q values of our algorithm converge to an interpretable set of values.
\begin{theorem}
\label{theorem:convergence}
\textit{For discount factor $\gamma \in [0,1)$, the Q values of Welfare Q-Learning converge.}
\end{theorem}
\begin{proof}
The proof follows from \textit{Banach’s Fixed-Point Theorem} \cite{banach1922operations}, which guarantees the existence and uniqueness of fixed-point of a contraction map on a complete metric space. The update step of Algorithm~\ref{alg:cap} can be seen as applying, in expectation, an operator on the Q-table. To apply the fixed-point theorem, we define a metric and show that this operator is a contraction. The Generalized Banach Fixed-Point Theorem therefore implies that Algorithm~\ref{alg:cap} will converge toward a unique fixed-point of this operator. Proofs for the two technical lemmas are provided in the full version \cite{https://doi.org/10.48550/arxiv.2212.01382}.

\begin{definition} Define $d$ on the space of Q-tables $\mathcal{Q}$ by
    $$
    d(\bm{Q},\bm{Q}') := \max_{\substack{s \in \mathcal{S}, a \in \mathcal{A}\\ i \in \{1,\dots,d\}} } \left|Q_i(s,a)) - (Q_i'(s,a))\right|.
    $$
\end{definition} 

\begin{lemma}
    $\langle \mathcal{Q}, d \rangle$ is a complete metric space.
\end{lemma}
    
Next, we define the \textit{optimality filter} $\mathcal{H}$.
    
\begin{definition}
\label{def:optimalityfilter}
    The optimality filter $\mathcal{H}$ is an operator defined as
    $$
    (\mathcal{H}\bm{\bm{Q}})(s) = \text{arg}_{\bm{Q}} \max_{a' \in \mathcal{A}} W(\bm{Q}(s, a')),
    $$
    where $\text{arg}_{\bm{Q}}$ takes the multi-objective value corresponding to the maximum, i.e., $\bm{\bm{Q}}(S, a'')$ such that
    $a'' \in \text{arg}\max_{a\in \mathcal{A}} W(\bm{\bm{Q}}(S,a'))$.
\end{definition}
    
We can define an \textit{optimality operator} $\mathcal{T}$ in terms of the optimal filter.
    
\begin{definition}
    The \textit{optimality operator} $\mathcal{T}$ is defined as
    $$
    (\mathcal{T}\bm{\bm{Q}})(s,a)= \bm{r}(s,a) + \gamma \mathbb{E}_{s' \sim \mathcal{P}(\cdot| s, a)}(\mathcal{H}\bm{\bm{Q}})(s').
    $$
\end{definition}
    
Note that in the algorithm, at each iteration, we sample from $\mathcal{P}(\cdot | s,a)$ to make an update. If the learning rate $\alpha$ satisfies the usual Robbins-Monro type conditions, namely $\sum \alpha = \infty$ and $\sum \alpha^2 < \infty$, the update at each iteration is, in expectation, applying the optimality operator $\mathcal{T}$. Thus, to show convergence, it suffices to show that iteratively applying $\mathcal{T}$ on any $Q$ leads to a unique $Q$-table.
    
\begin{lemma}[The optimality operator is a contraction]
    Let $\bm{\bm{Q}}, \bm{\bm{Q}'}$ be any two multi-objective \ $\bm{Q}$-value functions, then $d(\mathcal{T}\bm{\bm{Q}}, \mathcal{T}\bm{\bm{Q}'}) 
    \\ \le \gamma d(\bm{\bm{Q}}, \bm{\bm{Q}'})$, where $\gamma \in [0,1)$ is the discount factor of the underlying MOMDP.
\end{lemma} 
    
    \noindent
    Finally, since in our design, the distance $d$ is a well-defined metric, to prove convergence to a unique fixed point, we will use the Generalized Banach Fixed Point Theorem as in \cite{morl_convergence}.
    
\begin{lemma}[Generalized Banach Fixed-Point \cite{morl_convergence}]
    Given that $\mathcal{T}$ is a contraction mapping with Lipschitz coefficient $\gamma$ on the complete pseudo-metric space $\langle \mathcal{Q}, d \rangle$, then there exists $\bm{Q}^*$ such that
    $$
    \lim_{n\to \infty} d(\mathcal{T}^n\bm{\bm{Q}}, \bm{\bm{Q}^*}) = 0
    $$
    for any $\bm{\bm{Q}} \in \mathcal{Q}$.
\end{lemma}
    
It follows from the Lemmas that there exists $\bm{Q}^*$ such that
    $$
    \lim_{n\to \infty} d(\mathcal{T}^n\bm{\bm{Q}}, \bm{\bm{Q}^*}) = 0
    $$
    for any $\bm{\bm{Q}} \in \mathcal{Q}$. 
    
Note that if the distance $d$ between two tables is $0$, the tables are equal. In other words, iteratively applying the optimality operator $\mathcal{T}$ to a multi-objective Q-table will converge toward a unique Q-table. Since the update step in Algorithm~\ref{alg:cap} is applying $\mathcal{T}$ in expectation, the algorithm also converges toward a unique Q-table. This concludes the proof of Theorem \ref{theorem:convergence}.
\end{proof}
    
Note that the convergence result is not dependant on any particular welfare function $W$, but applies generally. 

Next, we provide an interpretation of the unique fixed point $\bm{Q}^*$ of our algorithm. Note that the Bellman optimality conditions are not satisifed for nonlinear welfare functions, precluding the typical interpretation of the optimal Q-table. We nevertheless show that a very similar interpretation can be given in which $\bm{Q}^*(s, a)$ provides a lower bound estimate on the discounted cumulative reward vector that is achievable after taking action $a$ in state $s$ and then optimizing for $W$.
    
\begin{definition}
\label{def:greedystationary}
Define a policy $\pi^*_Q$ given a Q-table as follows: at a given state $s$, let $\pi^*_Q(s_t) = a^*_t = \argmax_a W(\bm{Q}(s, a))$ and let $\bm{r}(s_t,a)$ be the immediate reward of performing action $a$ in state $s$. Then the value function corresponding to $\pi^*_Q(s)$ is
    $$
    \bm{V}_{\pi_Q^*(s)} = \mathbb{E}_{\tau \sim (\mathcal{P}, \pi^*)| s_1 = s} \sum_{t = 1}^\infty \gamma^{t-1} \bm{r}(s_t, a_t^*).
    $$
\end{definition}

\begin{lemma}[Interpreting the Fixed-Point]
\label{lemma:interpret}
    Let $\bm{Q}^*$ be the unique fixed point of the algorithm, i.e. $\bm{Q}^* = \mathcal{T}\bm{Q}^*$, then
    $$
    \bm{Q}^*(s, a) = \bm{r}(s,a) + \gamma \mathbb{E}_{s' \sim \mathcal{P}(\cdot | s,a)} \bm{V}_{\pi^*_{\bm{Q}^*}}(s').
    $$
\end{lemma}
    \begin{proof}
    Since $Q^*$ is the fixed point of the algorithm, $Q^* = \mathcal{T}Q^*$. Expanding this using the definition of $\mathcal{T}$, we get
    $$
    \bm{Q}^*(s,a) = \bm{r}(s,a) + \gamma \mathbb{E}_{s' \sim \mathcal{P}(\cdot | s,a)} \text{arg}_{\bm{Q}} \max_{a' \in \mathcal{A}(s')} W(\bm{Q}^*(s', a')).
    $$
    So it suffices to show that $\text{arg}_{\bm{Q}} \max_{a' \in \mathcal{A}} W(\bm{\bm{Q^*}}(s', a')) = \bm{V}_{\pi^*_{\bm{Q}^*}}(s')$. But expanding the LHS and RHS recursively and recalling $a_0^* = \argmax_{a'} W(\bm{Q}^*(s', a'))$ (see Definition~\ref{def:greedystationary}) as well as the definition of $\text{arg}_Q$ in Definition \ref{def:optimalityfilter}, we get
    $$
    \begin{aligned}
    &\text{arg}_{\bm{Q}} \max_{a' \in \mathcal{A}} W(\bm{\bm{Q^*}}(s', a')) = \bm{Q^*}(s', a_1^*) = \mathcal{T}[\bm{Q^*}(s', a_0^*)] \\
    &= \bm{r}(s',a_1^*) + \gamma \mathbb{E}_{s'' \sim \mathcal{P}(\cdot | s',a_1^*)} \text{arg}_{\bm{Q}} \max_{a'' \in \mathcal{A}} W(\bm{\bm{Q^*}}(s'', a''))
    \end{aligned}
    $$
    and
    $$
    \begin{aligned}
    &\bm{V}_{\pi^*_{\bm{Q}^*}}(s') = \mathbb{E}_{\tau \sim (\mathcal{P}, \pi^*_{\bm{Q}^*})| s_1 = s'} \sum_{t = 1}^\infty \gamma^{t-1} \bm{r}(s_t, a_t^*) \\
    &= r(s', a_1^*) + \gamma \mathbb{E}_{s'' \sim \mathcal{P}(\cdot | s',a')} \left[\mathbb{E}_{\tau \sim (\mathcal{P}, \pi^*_{\bm{Q}^*})| s_1 = s''} \sum_{t = 1}^\infty \gamma^{t-1} \bm{r}(s_t, a_t^*)\right] \\
    &= r(s', a_1^*) + \gamma \mathbb{E}_{s'' \sim \mathcal{P}(\cdot | s',a')} \left[\bm{V}_{\pi^*_{\bm{Q}^*}}(s'')\right].
    \end{aligned}
    $$
    So in turn, it suffices to show that $\text{arg}_{\bm{Q}} \max_{a'' \in \mathcal{A}} W(\bm{\bm{Q^*}}(s'', a''))$ 
    is equal to $\bm{V}_{\pi^*_{\bm{Q}^*}}(s'')$. Repeat this argument until the agent is in one of the terminal states $\bar{s}$.
    Note that $\bm{Q}(\bar{s}, \cdot) = 0$, i.e. Q-values for the terminal states are zero for all actions. So
    $$
    \text{arg}_{\bm{Q}} \max_{a \in \mathcal{A}} W(\bm{\bm{Q^*}}(\bar{s}, a)) = \bm{0} = \bm{V}_{\pi^*_{\bm{Q}^*}}(\bar{s}).
    $$
    This completes the proof of Lemma~\ref{lemma:interpret}.
    \end{proof}
    
Lemma \ref{lemma:interpret} implies that each entry $\bm{Q}^*(s, a)$ represents what the agent could actually expect to receive as total discounted future reward in expectation if one performs action $a$ in initial state $s$, then follows greedy stationary action selection using the Q-table. This is essentially the same interpretation as in traditional scalar Q-learning, except that there is no optimality guarantee of the result for nonlinear welfare optimization. 

Intuitively, these entries only serve as an estimate of the lower bound of total discounted rewards that are achievable in the future. Algorithm~\ref{alg:cap} combines these estimates with keeping track of the discounted cumulative rewards up to a given point in order to greedily optimize for the welfare of the sum of the two vectors. This is the essential intuition behind Algorithm~\ref{alg:cap}.

\section{Experiments}
\label{sec:experiments}
We run experiments under two simulated environments diagrammed in Figure~\ref{fig:environments}. \footnote{The implementation is available at \href{https://github.com/MuhangTian/Fair-MORL-AAMAS}{https://github.com/MuhangTian/Fair-MORL-AAMAS}} Our results demonstrate that (a) \textit{\textit{Welfare Q-Learning}} is effective in finding policies with high expected welfare compared with other baselines, (b) the rate of convergence depends on $n$, the dimensionality of the reward space, and (c) linear scalarization and mixture policies are generally inadequate for optimizing fair welfare functions. All results for all algorithms are obtained using an average of $\NSW$ and utilitarian welfare of $\bm{r}_{acc}$ for each episode over 50 runs. 
For all the experiments on both environments, each unit on the $x$-axis corresponds to an episode, which equals to $10000$ timesteps or action selections in the environment. The duration of a timestep is the same for all methods.

\begin{figure}[h]
\begin{subfigure}{0.5\linewidth}
     \centering
     \includegraphics[width=0.75\linewidth]{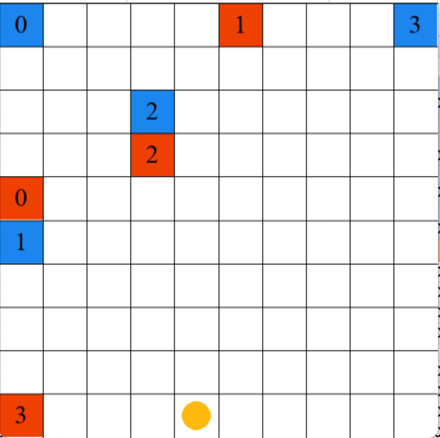}
     \caption{Taxi}
     \label{fig : env_taxi}
 \end{subfigure}%
 \begin{subfigure}{0.5\linewidth}
     \centering
     \includegraphics[width=0.75\linewidth]{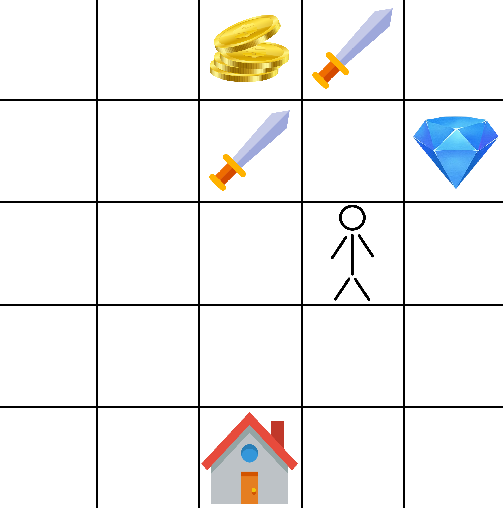}
     \caption{Resource Gathering \cite{Alegre+2022bnaic}}
\end{subfigure}
\caption{Simulated Environments}
\label{fig:environments}
\end{figure}

\subsection{Metrics, Methods, and Baseline Algorithms}
\begin{figure*}[h]
\begin{subfigure}{0.33\linewidth}
     \centering
     \includegraphics[width=1.05\linewidth]{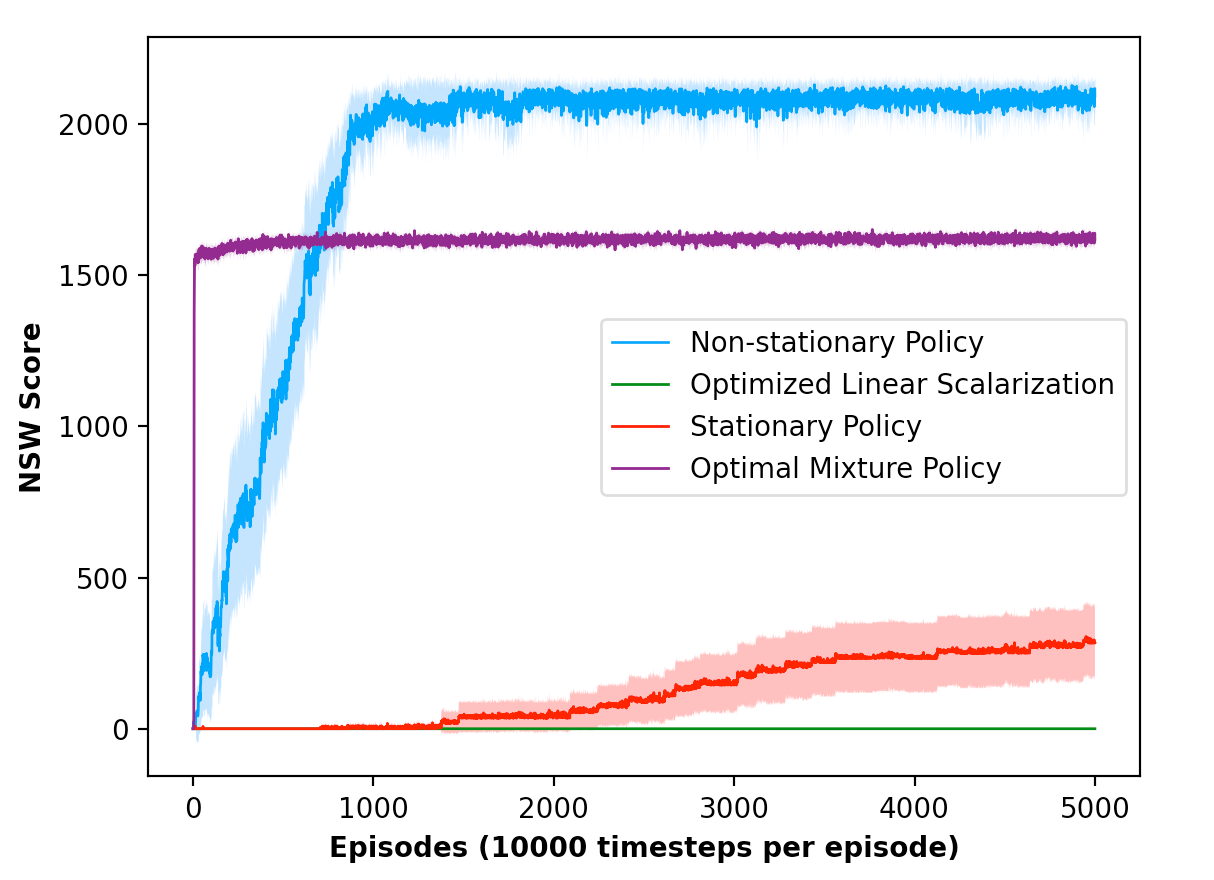}
     \caption{Online Performance (NSW)}
 \end{subfigure}%
 \begin{subfigure}{0.33\linewidth}
     \centering
     \includegraphics[width=1.05\linewidth]{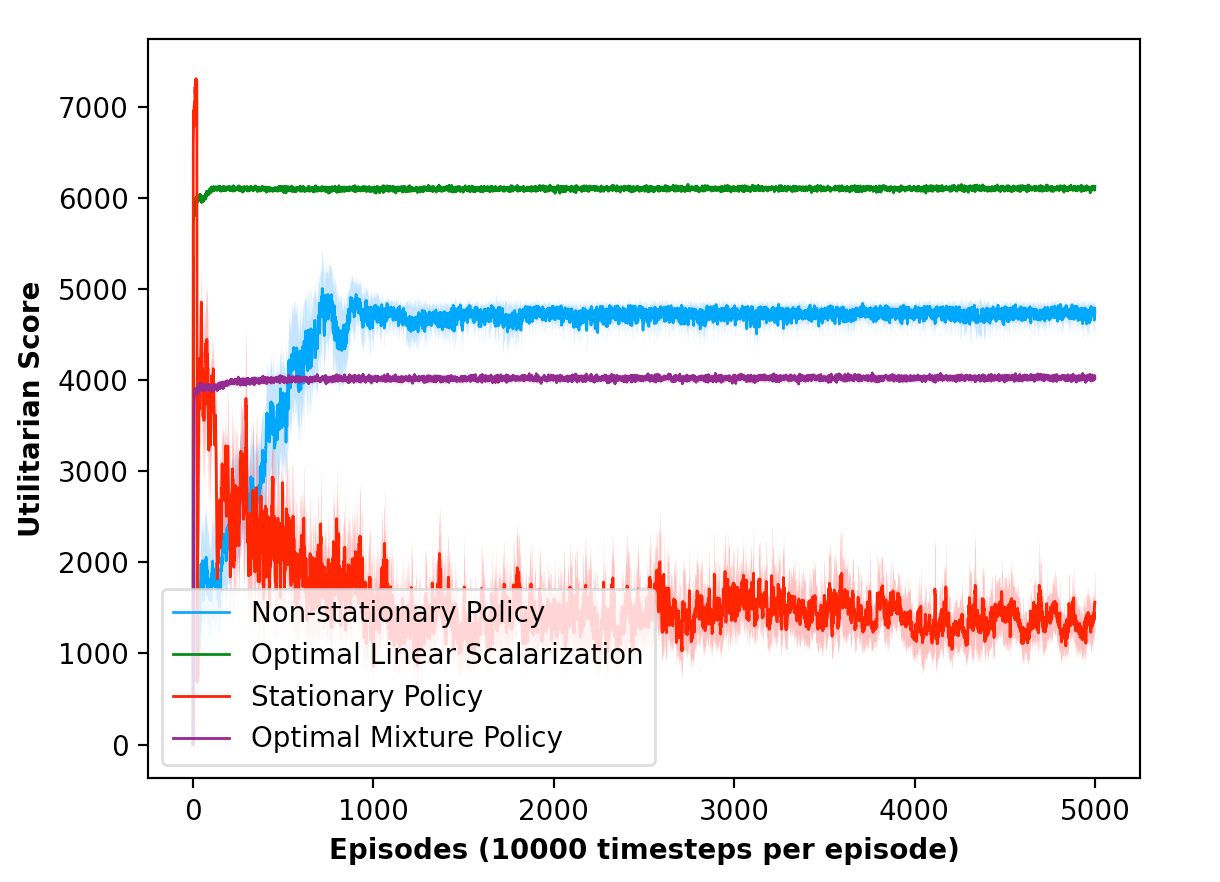}
     \caption{Online Performance (Utilitarian)}
\end{subfigure}%
\begin{subfigure}{0.33\linewidth}
    \centering
    \includegraphics[width=1.05\linewidth]{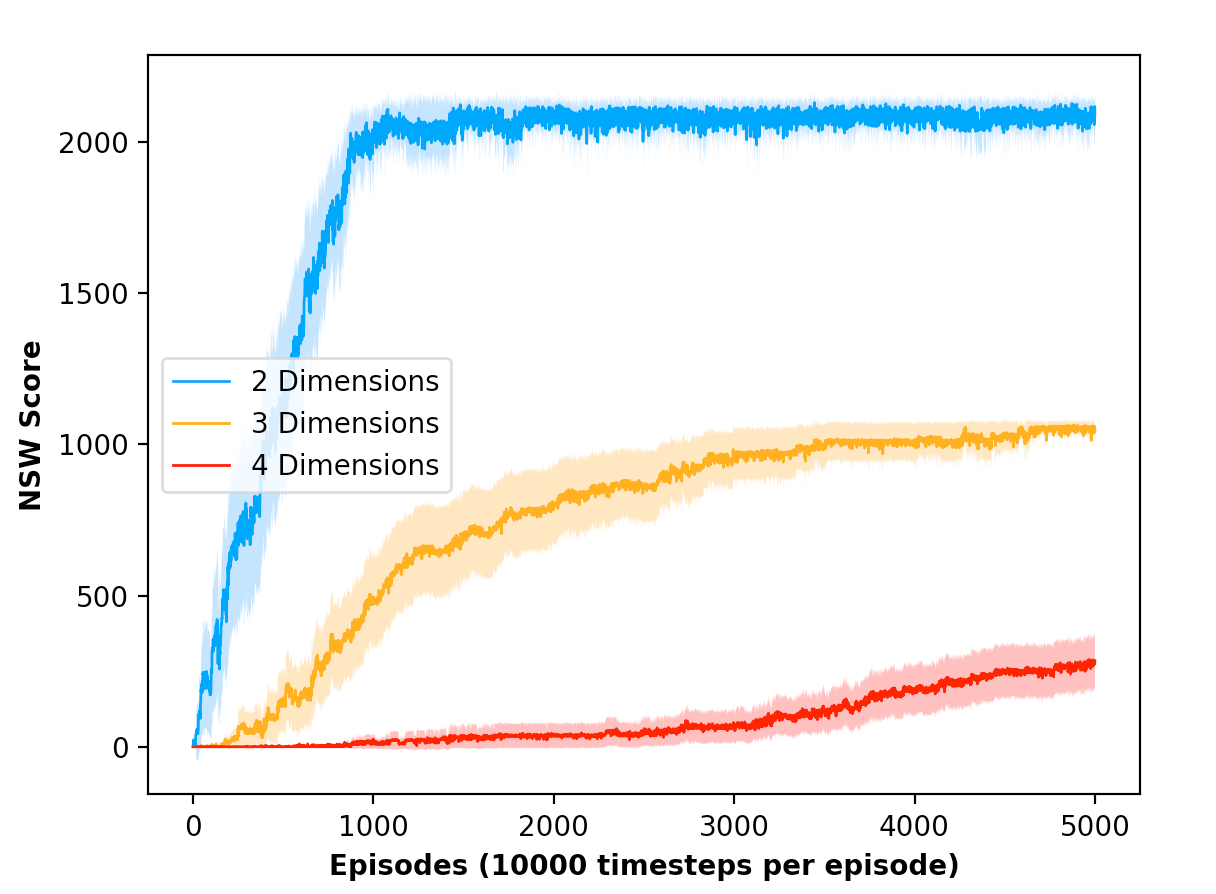}
    \caption{Taxi with Different Dimensions}
    \label{subfig:dimension}
\end{subfigure}
\caption{Experiment Results for Taxi Environment. Non-stationary Policy is Welfare Q-Learning}
\label{fig:resultstaxi}
\end{figure*}

All of our experiments attempt to optimize NSW (results for other welfare functions are provided in the full version \cite{https://doi.org/10.48550/arxiv.2212.01382}). We measure the $\NSW$ function on $\bm{r}_{acc}$ earned thus far. As elaborated in Section~\ref{alg:cap}, $\NSW$ satisfies all of our basic desiderata plus scale invariance, and is an intermediate welfare function between the extremes of egalitarian and utilitarian social welfare. For comparison, we also show utilitarian social welfare (the arithmetic mean of reward vectors) alongside NSW.

The geometric mean can be numerically unstable. In experiments we work with the log transform of $\NSW$. That is, instead of maximizing $\NSW(\bm{R}) = \left(\prod_i^n R_i\right)^{1/n}$, we equivalently maximize $\sum_i^n \ln(R_i + \lambda)$, where $\lambda > 0$ is included as a smoothing factor in case $R_i = 0$. Due to the nature of the $\NSW$ function, the $\NSW$ of rewards with negative elements is undefined, or alternatively can be defined as $-\infty$. Thus, our scope of exploration is restricted to policies that yield all non-negative accumulated rewards. 


\subsubsection{Baseline Algorithms} We compare our algorithm against three baselines. 
\begin{enumerate}[itemsep=2mm, leftmargin=4mm]
    \item \textit{Optimal Linear Scalarization.} A simple MORL technique is to apply linear scalarization on the Q-table \cite{van2013scalarized}. Given weights $\bm{w} \in \mathbb{R}^n$, where $\sum_{i=1}^n w_i = 1$ for $n$ objectives, let $SQ(s,a)=\sum_{i=1}^n w_i Q(s,a)_i$, where $Q(s,a)_i$ is the Q-value for $i^{th}$ objective. For each time step, $SQ(s,a)$ is treated by the algorithm as the objective to perform both action selection with $\epsilon$-$greedy$ and learning updates of $\bm{Q}(s,a)$. We chose the weights $\bm{w}$ that performed best on the NSW objective as determined by a grid search through combinations of $\bm{w}$. We include this algorithm as the baseline to demonstrate the limitations of linear scalarization on nonlinear objectives and show the importance of our nonlinear learning updates in Algorithm~\ref{alg:cap}.
    
    \item \textit{Stationary Policy.} Algorithm~\ref{alg:cap}, \textit{Welfare Q-Learning}, learns a particular Q-table corresponding to a (potentially) nonlinear welfare function, then performs non-stationary action selection. By contrast, we also show the results if one stationary $\epsilon$-greedy action selection on the same learned Q-table. That is, the stationary policy algorithm does not consider $\bm{r}_{acc}$ in its action selection. We include this algorithm to show the importance of non-stationary action selection.
    
    \item \textit{Optimal Mixture Policy.} \cite{vamplew2009constructing} proposed the idea of combining multiple Pareto Optimal base policies into a single mixture policy. We chose our base policies as those that optimize each dimension of the reward vector independently. The algorithm then uses one of these policies for $I$ time steps, then switches to the next. To determine the optimal value of $I$ for optimizing $\NSW$, we used a grid search. We use the resulting optimal value $I^*$. This baseline examines the effectiveness of intuitive approaches (combining optimal policies for each user) for optimizing fairness.
\end{enumerate}

\subsection{Taxi Environment}
\paragraph{Description.} Inspired by the taxi toy example problem for single-objective RL \cite{dietterich2000hierarchical}, we designed a multi-objective taxi environment \footnote{See detailed description of the environment in the full version \cite{https://doi.org/10.48550/arxiv.2212.01382}}. In this grid world, our agent is a taxi driver whose goal is to deliver passengers from their origins to their destinations. There are $n$ origin-destination pairs, one for each dimension of reward, and the agent earns reward in that dimension when dropping off a passenger from that origin-destination pair. There are an unlimited number of passengers for each origin-destination pair, but the taxi can only take one passenger at a time. This constraint enforces objectives to be conflicting, thus our agent's fairness performance becomes more important---it should provide its delivery service to each origin successfully and fairly over time, without ignoring origins that are more difficult to deliver (such as number 3 origin/destination pair in Figure \ref{fig : env_taxi}).

\begin{figure*}[h]
\begin{subfigure}{.33\linewidth}
\includegraphics[width=1.1\linewidth]{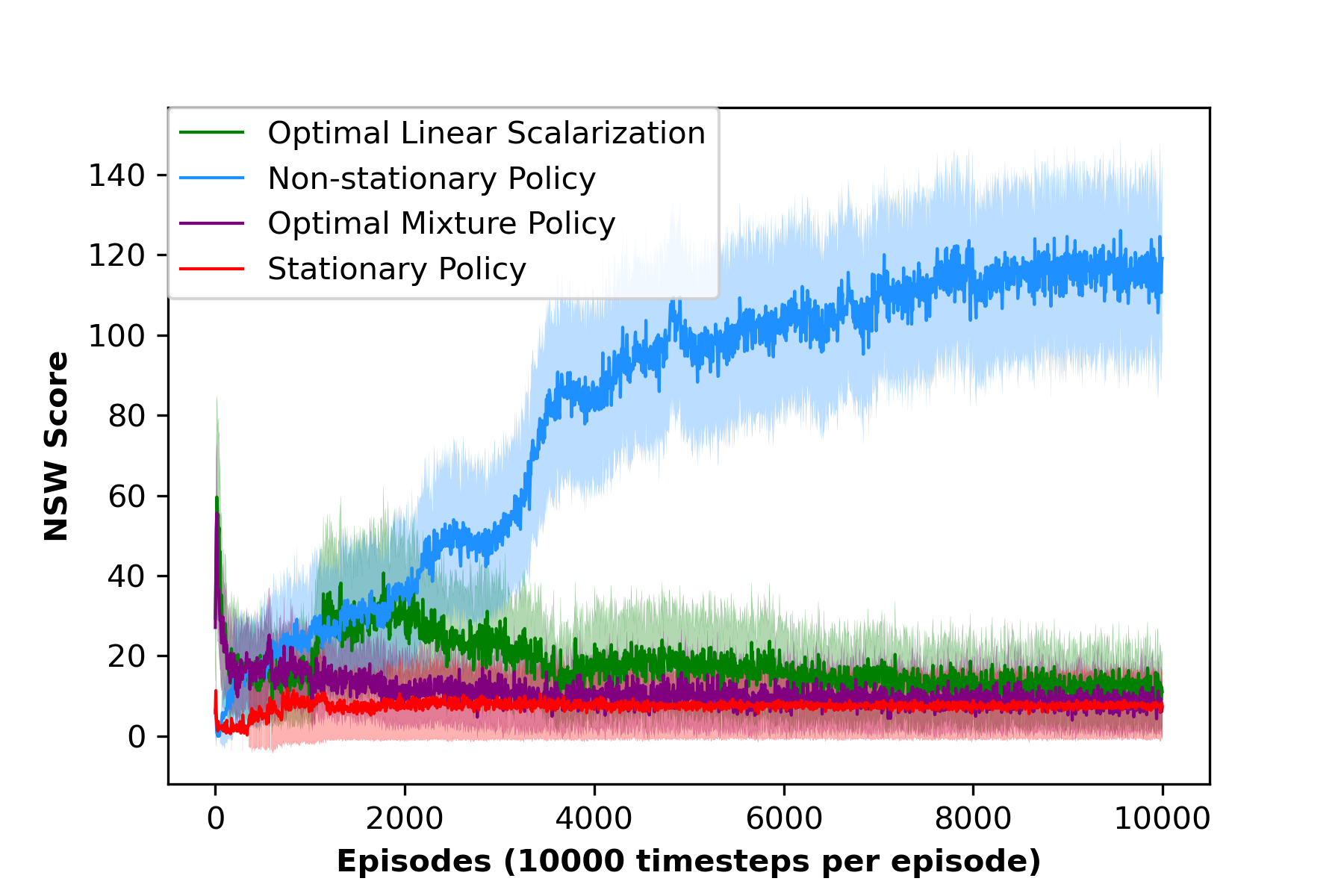}
\caption{Online Performance (NSW)}
    \label{fig:my_label}
\end{subfigure}%
\begin{subfigure}{.33\linewidth}
\centering  \includegraphics[width=1.1\linewidth]{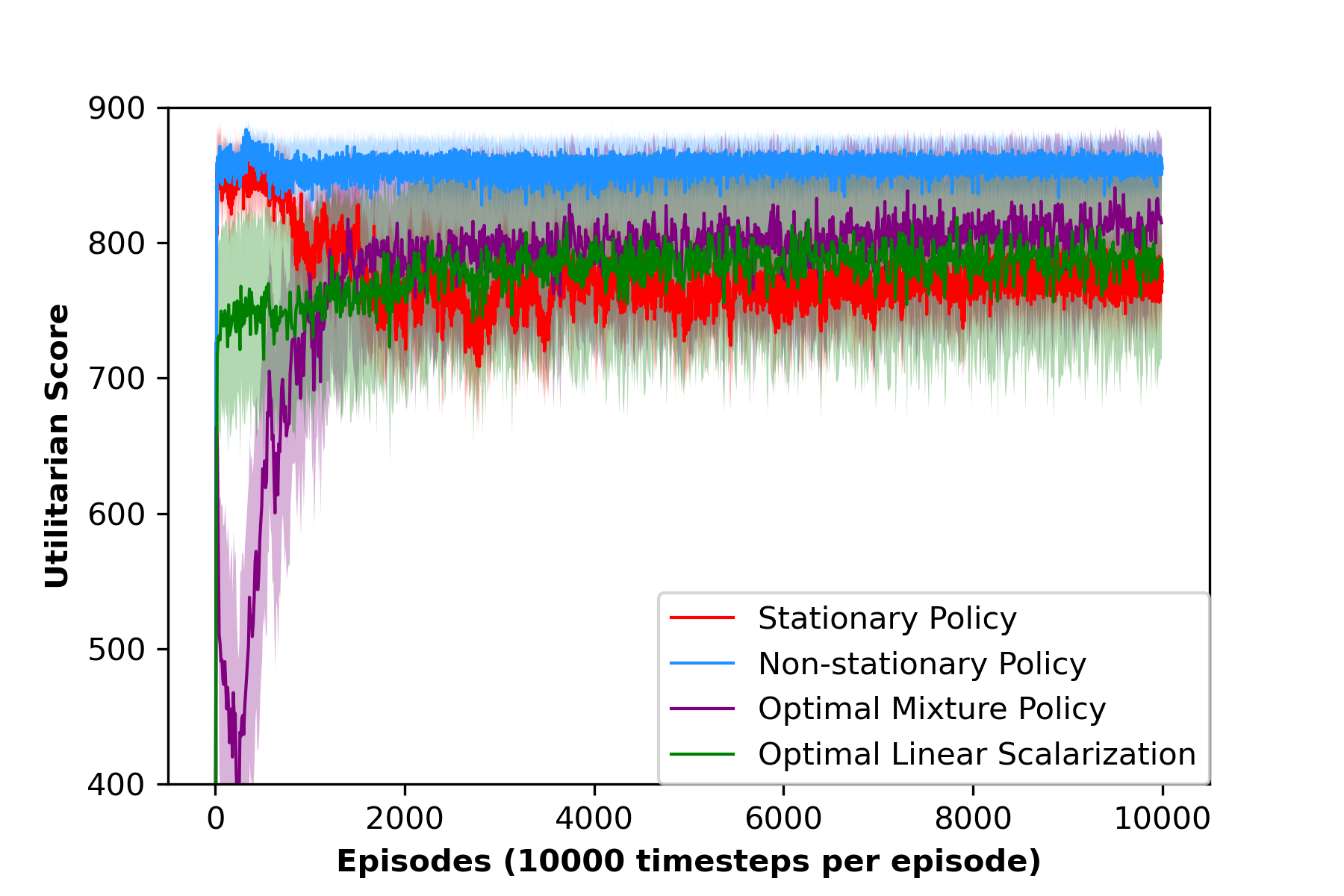}
    \caption{Online Performance (Utilitarian)}
    \label{fig:my_label}
\end{subfigure}%
\begin{subfigure}{.33\linewidth}
\centering 
 \includegraphics[width=0.7\linewidth]{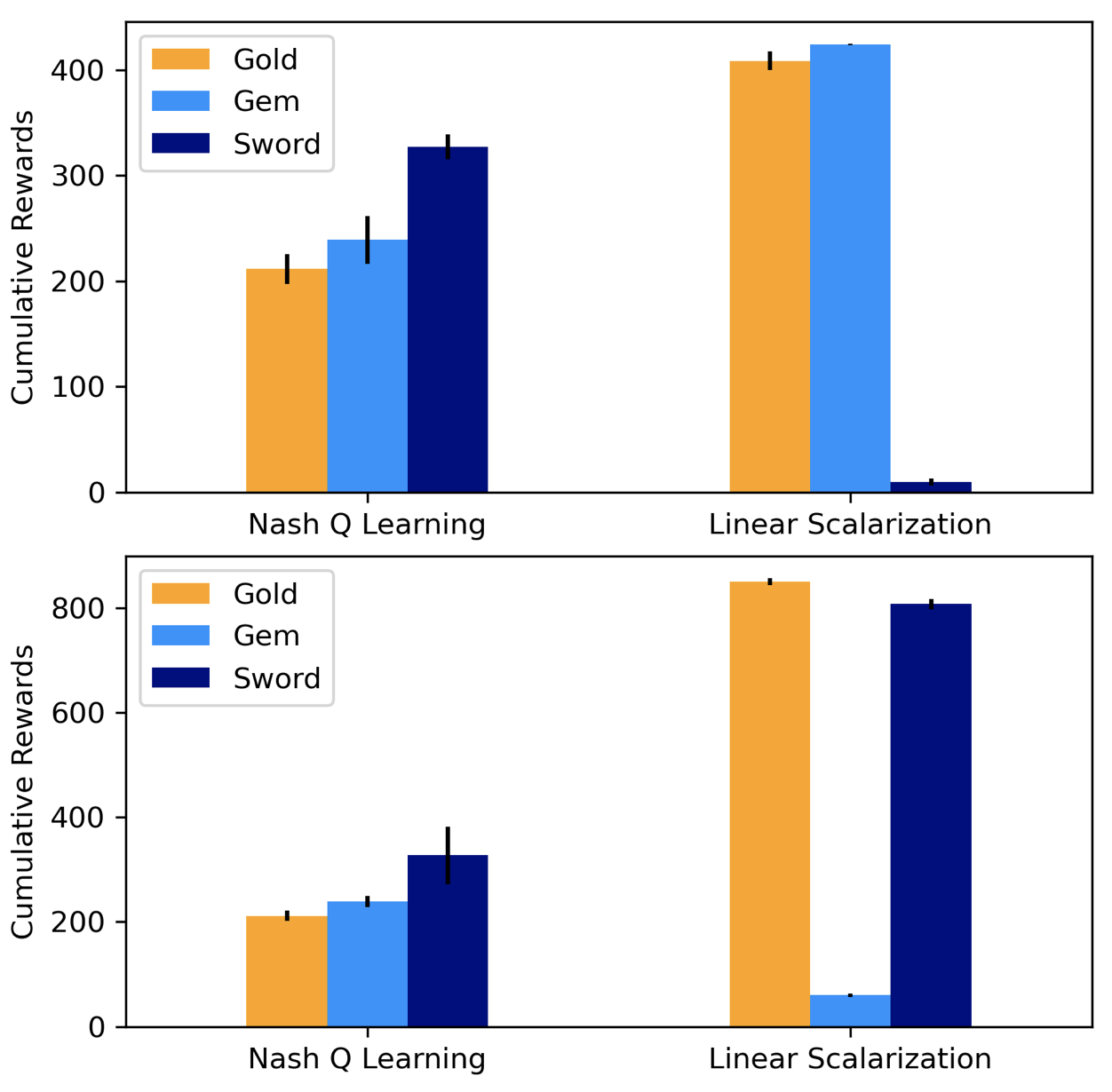}
 \caption{Distribution of Resources for \\ Equal (Top) and Scaled Rewards (Bottom)}
 \label{histograms}
\end{subfigure}
\caption{Experiment Results for RG Environment. Non-stationary Policy is Welfare Q-Learning}
\end{figure*}

\paragraph{Results.} Results are shown in Figure~\ref{fig:resultstaxi}. \textit{\textit{Welfare Q-Learning}} achieves the maximum average $\NSW$ score among all the algorithms, and still manages to achieve the second highest utilitarian score. We observe that our non-stationary policy outperforms the stationary policy on the same Q-table for both $\NSW$ and utilitarian score. Note that a stationary policy that optimizes NSW on this environment must essentially make a large loop always taking each origin-destination pair in turn, whereas a non-stationary policy can selectively optimize a single origin-destination pair for several time steps before switching to another pair. 

Linear scalarization has the lowest average $\NSW$ since there simply does not exist a set of weights that would produce accumulated rewards in all dimensions. It achieves highest utilitarian score since it favors to complete delivery for closest origin/destination pairs (such as index 2 pair in Figure \ref{fig : env_taxi}). The mixture policy performs generally well but slightly lower than that of \textit{\textit{Welfare Q-Learning}}, this is because an optimal fair policy in this environment does have the structure of alternating between optimizing on different dimensions at a time. This is not, however, true in general, as is seen in the next environment. Although the mixture policy converges quickly (each dimension independently is very easy to optimize), such performance is also subject to finding the optimal interval for the taxi environment $I^*$ ($227$ timesteps) via a search through the parameter space, which involves a computational cost not reflected by the figures.

For \textit{\textit{Welfare Q-Learning}}, we observe that there is an inverse correlation between the dimensionality of the reward space $n$ and the rate of convergence, as shown in Figure~\ref{subfig:dimension}. A possible explanation is that the increase in dimensionality increases the size of the Q-table, which is of size $|\mathcal{S}| \times |\mathcal{A}| \times n$, thus more updates are needed to converge.

\subsection{Resource Gathering (RG) Environment}

\paragraph{Description.} 
Building on the resource-collection domain \cite{barrett2008learning}, we modified the environment to include more complexity and randomness, and we also restructured it as a multi-objective setting. \footnote{See detailed description of the environment in the full version \cite{https://doi.org/10.48550/arxiv.2212.01382}} In this grid world, our agent collects three types of resources (gold, gem, and sword) spawned randomly at different locations which also disappear with a given probability. The dimensions of reward correspond to the resources. That is, gathering a given resource earns reward in that dimension and 0 the others. The goal of our agent is to collect as many resources as possible while maintaining a balance between difference types of resources.

\paragraph{Results.} 
\textit{\textit{Welfare Q-Learning}} achieves the maximum average NSW score among all the algorithms. Although linear scalarization achieves a reasonable utilitarian score, it fails to satisfy fairness, as it largely ignores one of the dimensions/resource types. The diminishing performance of the linear scalarization on NSW is presumably due to the algorithm increasingly coming to optimize for some dimensions and ignoring others. The mixture policy achieved the most comparable utilitarian score to the non-stationary policy, yet substandard NSW score. Unlike in the taxi environment, the optimal policy in RG is not characterized by optimizing each objective sequentially, which results in ignoring certain type of nearby resources.

An additional desirable property of optimizing NSW is \textit{scale invariance}, that the optimizing policy is invariant with respect to changes in the scale of a dimension of reward. We demonstrate this empirically by comparing the distribution over resources gathered using Nash Welfare Q-Learning (\textit{Welfare Q-Learning} with NSW as the welfare function) versus linear scalarization in Figure~\ref{histograms}. When all resources are worth the same amount of reward in their respective dimension (Figure \ref{histograms} Top) the NSW agent achieves a balanced distribution between the three types of resources, while the linear scalarized agent tends to neglect swords in the third dimension. To rectify this one might try to rescale the rewards (Figure \ref{histograms} Bottom) so that the value of swords is scaled to 50, but then the scalarized policy changes drastically to gather swords and ignore gems. However, the NSW agent is immune to scaling of this kind, retaining the same general distribution of reward in both cases. In other words, fair welfare (and especially NSW) optimization in a MOMDP may be more robust to specifications of the environment as compared to techniques based on linear scalarization. 

\balance
\section{Future work}
\label{sec:future}
While we know that exact optimization of, say, NSW is intractable, we do not know what provable approximation factors might be achievable in tabular MORL. Additionally, we observe a dependence between the dimensionality $n$ and the convergence rate of our algorithm, but we do not know whether this ``curse of dimensionality'' is fundamental to nonlinear welfare optimization in MORL. Finally, though even the tabular setting is challenging for nonlinear welfare optimization in MORL, we believe the intuition of non-stationary action selection coupled with nonlinear learning updates can be extended to the function approximation setting and combined with deep neural network representations.

\bibliographystyle{ACM-Reference-Format} 
\bibliography{refs.bib}


\end{document}


\noindent{\huge \scshape Supplementary Material}
\thispagestyle{empty}
\section{Environment Descriptions}
\subsection{Taxi}
\begin{figure}[h]
    \centering
    \includegraphics[width=0.2\linewidth]{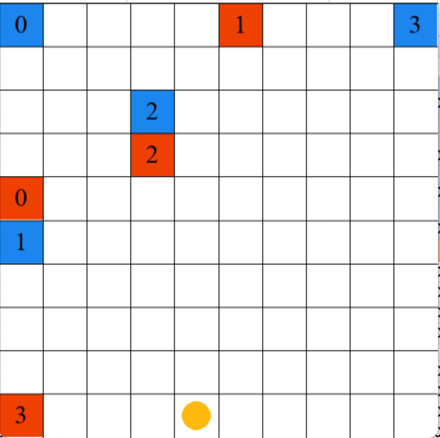}
    \caption{Visualization of Taxi grid world, orange circle is the taxi, origins are blue squares, destinations are red squares, with numbers indicating the corresponding origin and destination pairs}
    \label{fig:my_label}
\end{figure}
\begin{enumerate}
    \item State space: contains information about location of taxi on the grid, whether there is passenger in taxi, destination of the passenger in taxi
    \item Action space: move north, south, east, west, pick passenger, drop passenger
    \item Reward function: 
\end{enumerate}
\begin{equation*}
    \bm{R}_t = \begin{cases}
    \bm{0}  & \text{if $a_t \in \{$north, south, east, west $\}$ or $a_t$ is a valid pick or drop} \\
    \bm{-10} & \text{if invalid pick or drop is performed}\\
    R_i = 30, R_j = 0, \forall j \neq i & \text{if passenger from origin $i$ in taxi is dropped correctly to its destination}\\
    \end{cases}
\end{equation*}

In this environment, the agent is a taxi driver who is trying to deliver multiple passengers from their origins to their destinations. For simplicity, we assume there are infinite number of passengers at each origin, and the task is modeled as a \textit{continuing} task and therefore has no \textit{terminal state}. The state space contains information about location of taxi, whether there is passenger currently in taxi, as well as destination of the passenger in taxi. Our agent has six actions: drive north, south, east, west, and pick and drop passenger. The dimension of the objectives is the number of origin and destination pairs, which can be decided arbitrarily as a parameter in the environment. At each time step, the agent receives a reward of $\mathbf{r}=\mathbf{0}$ for movement, $\mathbf{r}=\mathbf{-10}$ for illegal action (dropping or picking at incorrect locations), and a reward of 30 at the dimension of the origin location for correct delivery, 0 for others. We also restricted the taxi to carry only one passenger at a time. This constraint enforces objectives to be conflicting, where delivery of one passenger implies ignoring the others. Under this particular setting, the agent's fairness performance becomes more important. It should provide its delivery service to each location successfully and fairly over time within each episode, without ignoring certain locations.

\subsection{Resource Gathering (RG)}

\begin{figure}[h]
    \centering
    \includegraphics[width=0.2\linewidth]{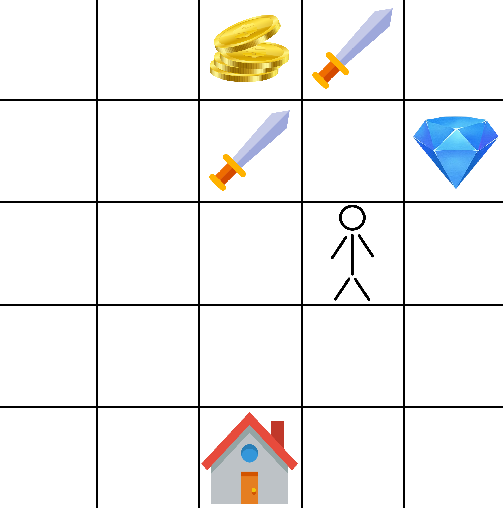}
    \caption{Visualization of RG grid world}
    \label{fig:my_label}
\end{figure}
The RG domain is a 5 $\times$ 5 grid world where the agent collects three types of resources (gold, gem, and sword) spawned randomly at different locations with a 99.99\% stochastic probability and disappearing with a small probability of 10\%. With a probability of 99.99\%, a new map is generated randomly indicating the newly updated locations of the resources. Our agent has four actions: traveling up, down, left and right in the four cardinal directions. The reward encodes the value achieved at each resource type (determined by the quantity of resources and the reward value of the resource type). Under equal rewards, the agent receives a reward $R=10$ in one dimension for gathering one resource in the corresponding dimension. We trained the agent with non-stationary action selection in discounted continuous task, and evaluated the agent with 10000 steps over resources of equal rewards and scaled rewards, and recorded accumulated rewards for each resource type. The goal of our agent is to collect as more resources as possible while maintaining a balance between difference types of resources.

\section{Experimental Results for Other Welfare Functions}
\thispagestyle{empty}
We run experiments for \textit{Welfare Q-Learning} based on $P$-welfare and \textit{egalitarian} welfare functions. We choose a range of values of $P$ between $[-1,1]$, and recorded each of its performance with $\NSW$ and utilitarian score.
\begin{figure}[h]
\centering
\begin{subfigure}{0.5\linewidth}
\includegraphics[width=1\linewidth]{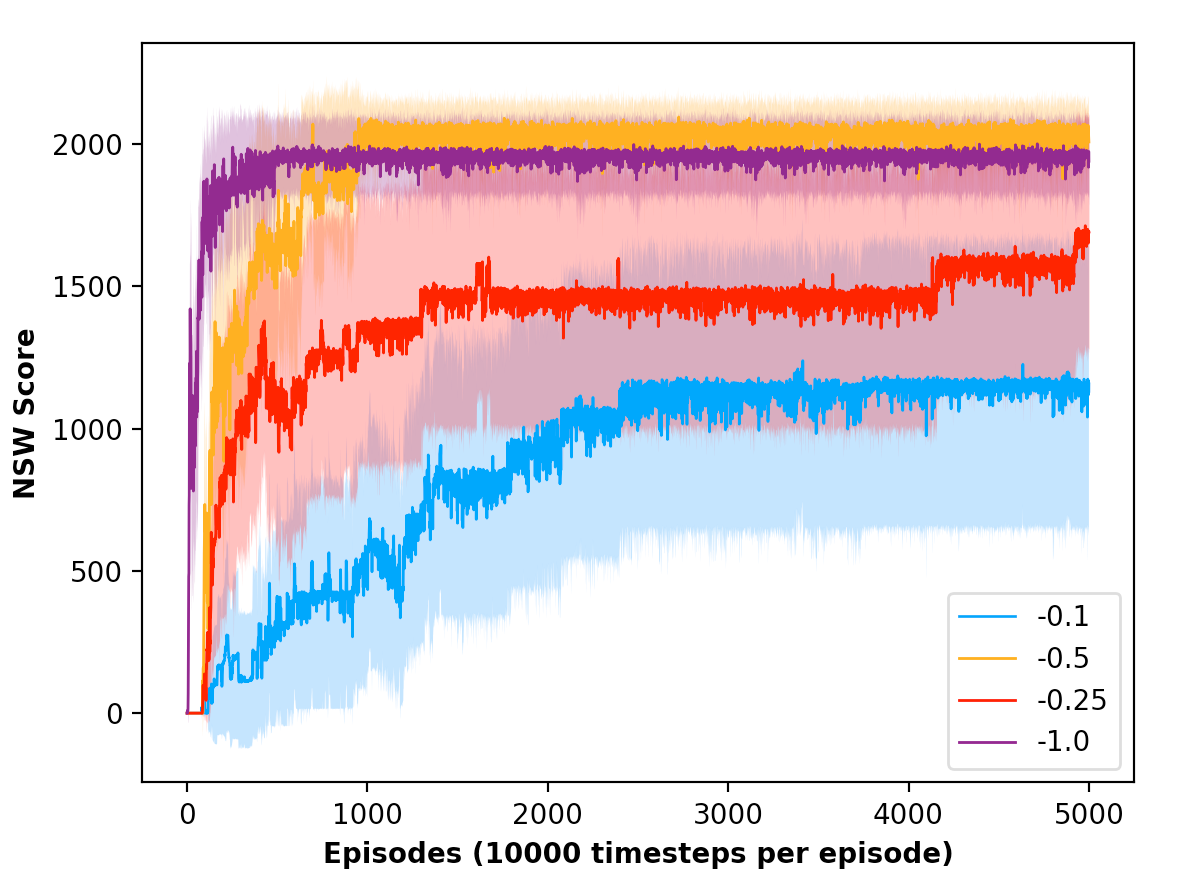}
\end{subfigure}%
\begin{subfigure}{0.5\linewidth}
\includegraphics[width=1\linewidth]{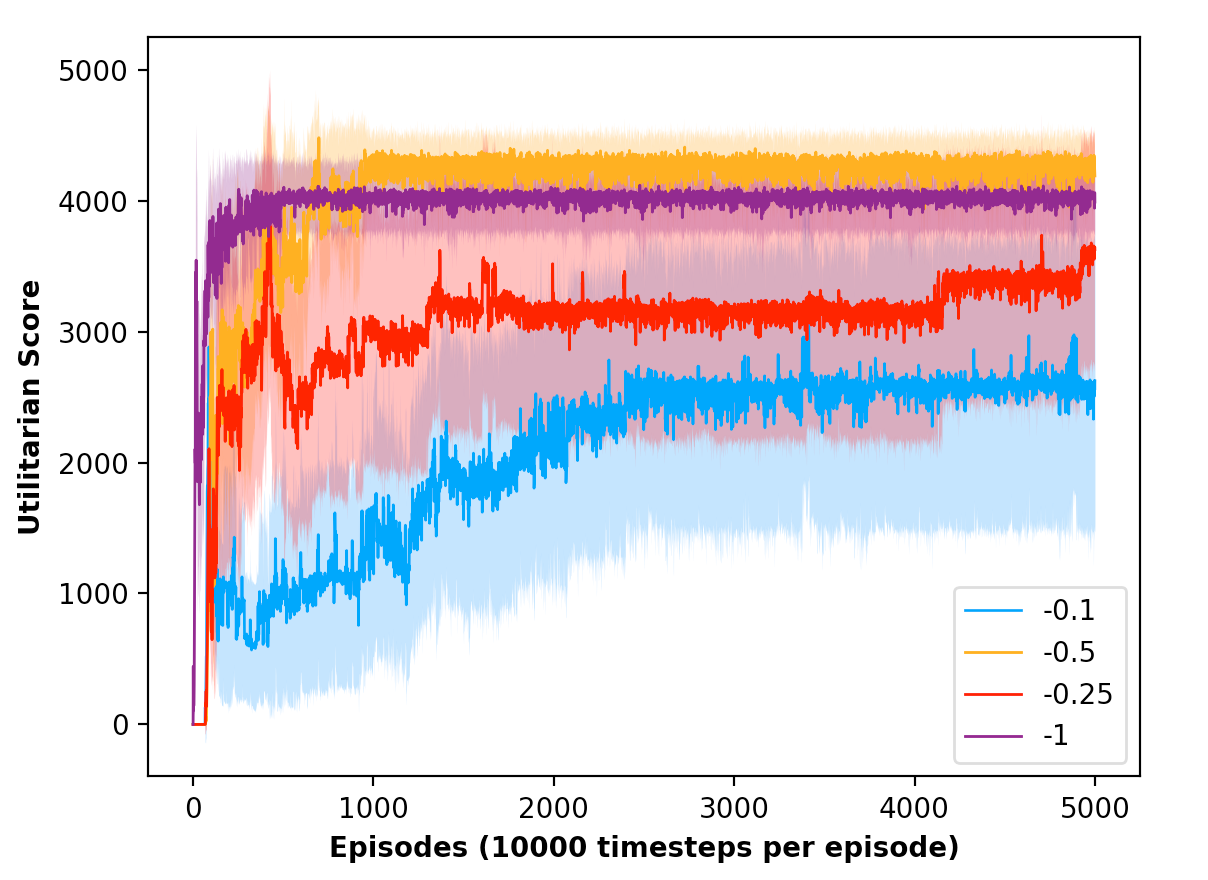}
\end{subfigure}
\begin{subfigure}{0.5\linewidth}
\includegraphics[width=1\linewidth]{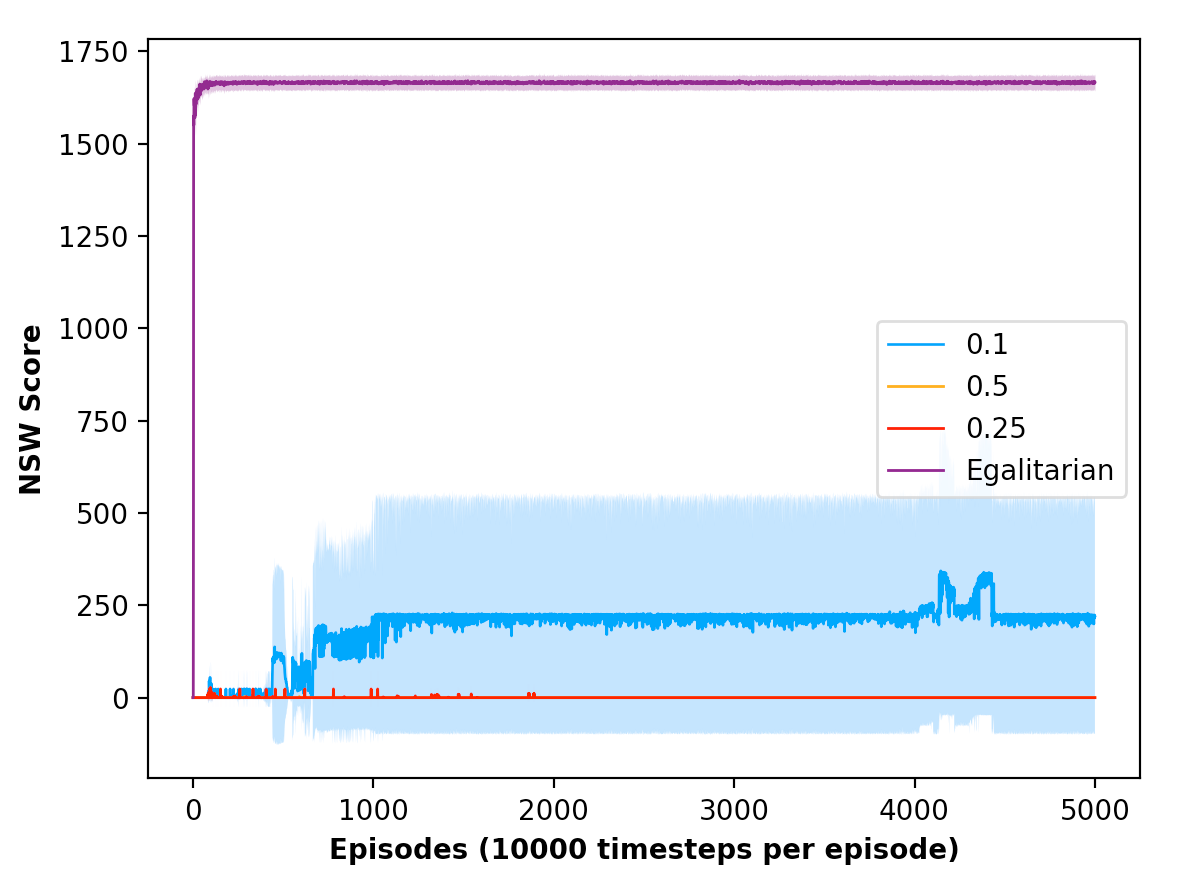}
\end{subfigure}%
\begin{subfigure}{0.5\linewidth}
\includegraphics[width=1\linewidth]{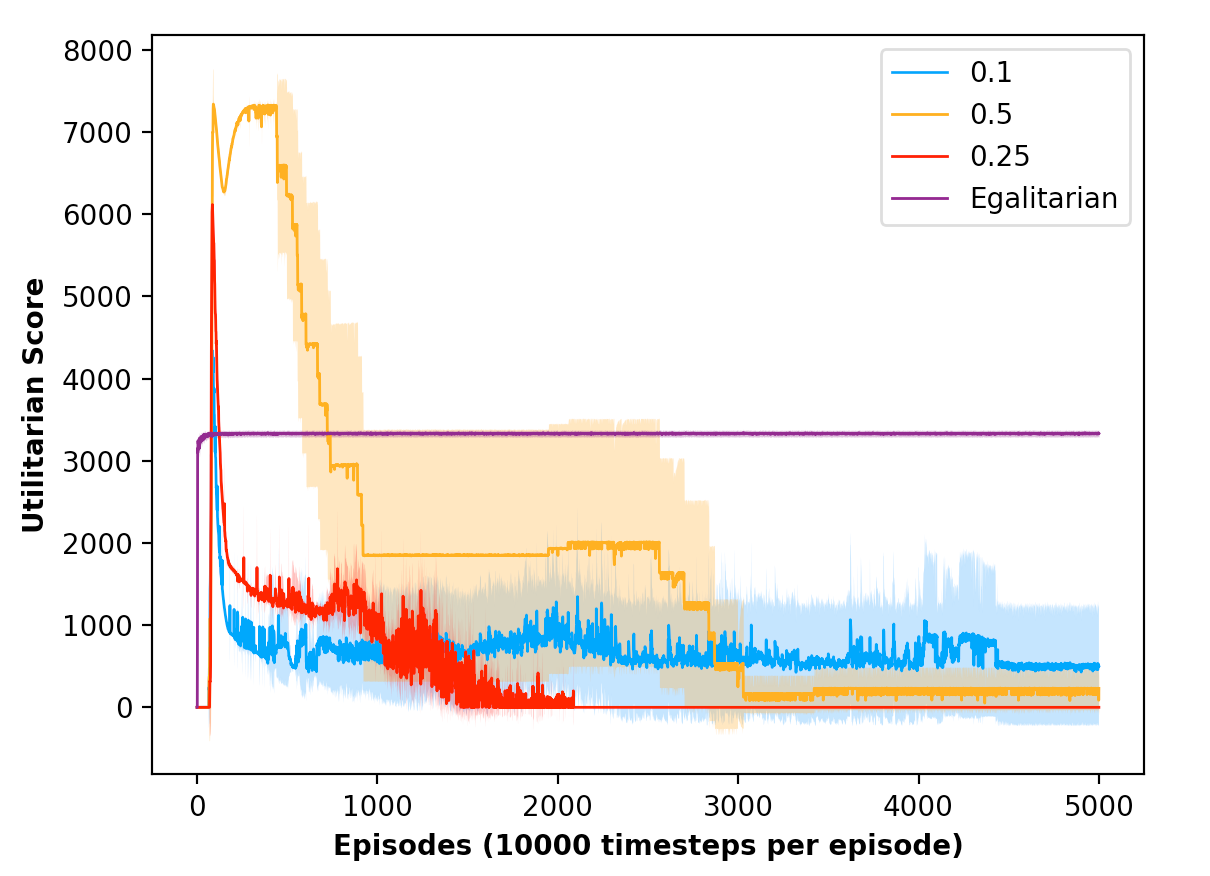}
\end{subfigure}
\caption{Experimental Results for Other Welfare Functions}
\end{figure}

\section{Proof of Convergence for Welfare $Q$-learning}
In this section, we provide convergence proof for our multi-objective algorithm. The proof is based on the well-known Banach’s Fixed-Point Theorem, which guarantees the existence and uniqueness of fixed-point of a contraction map on a complete metric space. Therefore, generalizing this theorem a bit, we can imagine all value functions of reinforcement learning are in some metric space, and finding the optimal value or policy is to find the fixed point of a certain contraction on that space. To do this, we i) define a well-defined metric on the space of $\bm{Q}$-functions; ii) show that the optimal operator is a contraction; and finally iii) apply the Generalized Banach Fixed-Point Theorem.

\begin{theorem}
\label{theorem:convergence}
\textit{For discount factor $\gamma \in [0,1)$, the Q values of Welfare Q-Learning converge.}
\end{theorem}
\begin{proof}\hfill

\textbf{Definition 3.1} Define a value metric $d$ on the space of Q-tables $\mathcal{Q}$ by
$$
d(\bm{Q},\bm{Q}') := \max_{\substack{s \in \mathcal{S}, a \in \mathcal{A}\\ i \in \{1,\dots,d\}} } \left|Q_i(s,a)) - (Q_i'(s,a))\right|
$$ 
\end{proof}
\thispagestyle{empty}
\begin{lemma} The value metric $d$ is a well-defined metric.
    \begin{proof}
    $d(\bm{Q},\bm{Q}') = 0 \iff |Q_i(s,a) - Q_i'(s,a)| = 0 \iff \bm{Q}(s,a) = \bm{Q}'(s,a)$ for all $s, a$. So positivity holds for $d$. If we choose $\{s^*, a^*, i^*\} = \argmax d(\bm{Q},\bm{Q}'')$, then
    $$
    \begin{aligned}
    \hspace{1cm} &d(\bm{Q},\bm{Q}') + d(\bm{Q},\bm{Q}'') \\
    &\ge |Q_{i^*}(s^*,a^*) - Q_{i^*}'(s^*,a^*)| + |Q'_{i^*}(s^*,a^*) - Q''_{i^*}(s^*,a^*)| \\
    &\ge |Q_{i^*}(s^*,a^*) - Q''_{i^*}(s^*,a^*)| = d(\bm{Q},\bm{Q}')
    \end{aligned}
    $$
    So triangle inequality holds for $d$. It is easy to verify from the definiton that $d(\bm{Q},\bm{Q}) = 0$ and $d(\bm{Q},\bm{Q}') = d(\bm{Q}', \bm{Q})$. So $d(\cdot)$ is indeed a well-defined metric.
    \end{proof}
    \end{lemma}
    
    \noindent
    \textbf{Remark.} It is easy to show that metric space $(\mathcal{Q}, d)$ is complete. \\
    Next, similar to scalarized $\bm{Q}$-learning, we design an optimality filter $\mathcal{H}$ defined by \\
    
    \noindent
    \textbf{Definition 5.2}
    $$
    (\mathcal{H}\bm{\bm{Q}})(s) := \text{arg}_{\bm{Q}} \max_{a' \in \mathcal{A}} W(\bm{Q}(s, a'))
    $$
    where $\text{arg}_{\bm{Q}}$ takes the multi-objective value corresponding to the maximum, i.e., $\bm{\bm{Q}}(S, a'')$ such that
    $a'' \in \text{arg}\max_{a\in \mathcal{A}} W(\bm{\bm{Q}}(S,a'))$, and $W$ is a welfare function of interest. \\
    
    \noindent
    Using the definition of the optimality filter, we can then write the optimality operator $\mathcal{T}$ in terms of the optimal filter: \\
    
    \noindent
    \textbf{Definition 5.3}
    $$
    (\mathcal{T}\bm{Q})(s,a):= \bm{r}(s,a) + \gamma \mathbb{E}_{s' \sim \mathcal{P}(\cdot| s, a)}(\mathcal{H}\bm{Q})(s')
    $$
    
    \noindent
    \textbf{Remark.} Note that in the algorithm, at each iteration, we sample from $\mathcal{P}(\cdot | s,a)$ to make an update. If the learning rate $\alpha$ satisfies the usual Robbins-Monro type conditions, namely $\sum \alpha = \infty$ and $\sum \alpha^2 < \infty$, the update at each iteration is, in expectation, applying the optimality operatior $\mathcal{T}$. Thus, to show convergence, it suffices to show that iteratively applying $\mathcal{T}$ on any $Q$ leads to a unique $Q$-table.
    
    \begin{lemma} \textbf{(Optimal Operator is a Contraction).} Let $\bm{\bm{Q}}, \bm{\bm{Q}'}$ be any two multi-objective $\bm{Q}$-value functions, then $d(\mathcal{T}\bm{\bm{Q}}, \mathcal{T}\bm{\bm{Q}'}) \le \gamma d(\bm{\bm{Q}}, \bm{\bm{Q}'})$, where $\gamma \in [0,1)$ is the discount factor of the underlying MOMDP.
    \begin{proof}
    Without loss of generality, we assume $\max_{\substack{a \in \mathcal{A}} } Q_i(s,a)) \ge \max_{\substack{a \in \mathcal{A}} }Q_i'(s,a)$ for some state $s$ and component $i$ of interest.
    Expand the expression of $d(\mathcal{T} \bm{Q}, \mathcal{T} \bm{Q}')$, we have
    \begin{align*}
        d(\mathcal{T} \bm{Q}, \mathcal{T} \bm{Q}') &= \max_{\substack{s \in \mathcal{S}, a \in \mathcal{A} \\ i \in \{1,\dots,d\}}} \left|(\mathcal{T}\bm{Q})_i(s,a)) - (\mathcal{T}\bm{Q}')_i(s,a))\right| \\
        &= \max_{\substack{s \in \mathcal{S}, a \in \mathcal{A} \\ i \in \{1,\dots,d\}}} \Bigg| \gamma \cdot E_{s' \sim P(\cdot | s,a)}(\mathcal{H}\bm{Q})_i(s') -\gamma \cdot E_{s' \sim P(\cdot | s,a)}(\mathcal{H}\bm{Q}')_i(s')) \Bigg|\\
        &\le \gamma \max_{\substack{s' \in \mathcal{S}\\ i \in \{1,\dots,d\}}} \Bigg|\text{arg}_{\bm{Q}} \left[\max_{a' \in \mathcal{A}} W(\bm{Q}(s',a'))\right]_i - \text{arg}_{\bm{Q}} \left[\max_{a'' \in \mathcal{A}} W(\bm{Q}'(s',a'')) \right]_i\Bigg| \tag{1}\\
        \end{align*}
     where $(1)$ is due to $|E[\cdot]| \le E[| \cdot |] \le \max | \cdot |$. According to our assumption, let $a'$ be the action chosen to maximize the value of $Q_i(s,a')$ for some state $s$ and component $i$ of interest, then we have     
        \begin{align*}
        & d(\mathcal{T} \bm{Q}, \mathcal{T} \bm{Q}') \\
        &\le \gamma \max_{\substack{s' \in \mathcal{S}\\ i \in \{1,\dots,d\}}} \Bigg|\mathrm{arg}_{\bm{Q}} \left[\max_{a' \in \mathcal{A}} W(\bm{Q}(s',a'))\right]_i -\mathrm{arg}_{\bm{Q}} \left[\max_{a'' \in \mathcal{A}} W(\bm{Q}'(s',a'')) \right]_i\Bigg| \\
        &\le \gamma \max_{\substack{s' \in \mathcal{S}\\ i \in \{1,\dots,d\}}} \Bigg|\mathrm{arg}_{\bm{Q}} \left[ W(\bm{Q}(s',a'))\right]_i -\mathrm{arg}_{\bm{Q}} \left[ W(\bm{Q}'(s',a'))\right]_i  + \mathrm{arg}_{\bm{Q}} \left[ W(\bm{Q}'(s',a'))\right]_i - \mathrm{arg}_{\bm{Q}} \left[\max_{a'' \in \mathcal{A}} W(\bm{Q}'(s',a'')) \right]_i\Bigg| \\
        &\le \gamma \max_{\substack{s' \in \mathcal{S}\\ i \in \{1,\dots,d\}}} \Bigg| \mathrm{arg}_{\bm{Q}} \left[ W(\bm{Q}(s',a'))\right]_i - \mathrm{arg}_{\bm{Q}} \left[ W(\bm{Q}'(s',a'))\right]_i \Bigg| \tag{2} \\
        &\le \gamma \max_{\substack{s' \in \mathcal{S}, a' \in \mathcal{A}\\ i \in \{1,\dots,d\}} } \left|Q_i(s',a')) - (Q_i'(s',a'))\right| \tag{3} = \gamma d(\bm{Q},\bm{Q}')
        \end{align*}
     $(2)$ arises from the w.l.o.g. assumption that \\ $Q_i(s', a') - \max_{a''} Q_i(s', a'') \ge 0$. Thus, the whole expression in $|\cdot|$ is nonnegative and $Q_i(s', a') - Q_i(s', a') \ge 0$. We can discard the last two terms since $Q'_i(s', a') \le \max_{a''} Q_i(s', a'')$. $(3)$ is due to $\max_{s', i} f(s', a')\le \max_{s', a''} f(s', a'')$ holds for any $a'$ and $f(\cdot)$. This completes our proof that $\mathcal{T}$ is a contraction.
    \end{proof}
    \end{lemma}
    \thispagestyle{empty}
    \noindent
    Finally, since in our design, the distance $d$ is a well-defined metric, to prove convergence to a unique fixed point, we will use the Generalized Banach Fixed Point Theorem. \\
    \begin{lemma}\textbf{(Generalized Banach Fixed-Point Theorem)}
    
    Given that $\mathcal{T}$ is a contraction mapping with Lipschitz coefficient $\gamma$ on the complete pseudo-metric space $\langle \mathcal{Q}, d \rangle$, then there exists $\bm{Q}^*$ such that
    $$
    \lim_{n\to \infty} d(\mathcal{T}^n\bm{\bm{Q}}, \bm{\bm{Q}^*}) = 0
    $$
    for any $\bm{\bm{Q}} \in \mathcal{Q}$.
    \end{lemma}
    
    \noindent
    Since our metric $d$ is a well-defined metric by Lemma 2 and therefore $\langle \mathcal{\bm{Q}}, d \rangle$ is a complete metric space, which is also a complete pseudo-metric space. Also, by Lemma 3, $\mathcal{T}$ is a contraction. So it follows from Lemma 4 that there exists $\bm{Q}^*$ such that
    $$
    \lim_{n\to \infty} d(\mathcal{T}^n\bm{\bm{Q}}, \bm{\bm{Q}^*}) = 0
    $$
    for any $\bm{\bm{Q}} \in \mathcal{Q}$. In other words, iteratively applying optimal operator $\mathcal{T}$ on any multi-objective Q-table, the algorithm will terminate with a unique table. Since in Welfare $Q$-learning, the update at each iteration is in expectation applying $\mathcal{T}$, the algorithm is convergent. This concludes the proof of Theorem 3.1.

    $$
    d(\bm{Q},\bm{Q}') := \max_{\substack{s \in \mathcal{S}, a \in \mathcal{A}\\ i \in \{1,\dots,d\}} } \left|Q_i(s,a)) - (Q_i'(s,a))\right|
    $$